\documentclass[a4paper,onecolumn,11pt]{quantumarticle}
\pdfoutput=1
\usepackage{tabularx} 
\usepackage{tabularx,makecell}
\usepackage[utf8]{inputenc}
\usepackage{amsmath,amssymb,amsthm,mathtools}
\usepackage{hyperref}
\usepackage{graphicx}
\usepackage{tikz}
\usepackage{etoolbox}
\usepackage{newunicodechar}
\usepackage{tcolorbox}
\newunicodechar{ε}{$\epsilon$}
\usetikzlibrary{arrows.meta,positioning,calc,fit}
\usepackage{enumitem}
\usepackage[ruled,vlined,linesnumbered]{algorithm2e}
\usepackage{booktabs}
\usepackage{pgfplots}
\usepackage{pgfplotstable}
\usepackage{caption}
\usepackage{filecontents} % built-in in newer LaTeX; harmless if already present
\pgfplotstableset{col sep=comma}

\usepackage{hyperref}
\usepackage{float}              % for [H]
\usepackage[section]{placeins}  % for \FloatBarrier

\makeatletter
\def\fps@figure{htbp}
\def\fps@table{htbp}
\makeatother

\usepackage{float}              % for [H]
\usepackage{tabularx,makecell}  % wrapping tables
\usepackage[section]{placeins}  % keep floats inside their section
\usepackage{booktabs}
\usepackage{caption}
\usepackage{caption,subcaption}

\pgfplotsset{compat=1.17}

\newtheorem{definition}{Definition}
\newtheorem{lemma}{Lemma}
\newtheorem{theorem}{Theorem}
\newtheorem{proposition}{Proposition}

\newtheorem{corollary}{Corollary}
\newtheorem{remark}{Remark}

\title{Curved Boolean Logic: A Contextual Generalization of Propositional Logic with Algorithmic Consequences}

\author{Maximilian R.~P.~von~Liechtenstein}
\affiliation{Independent Researcher}
\date{}

\begin{document}

% ===== Portable CSVs written at compile time (arXiv-safe) =====
% pgfplotstable expects these files and columns:
%   cbl_alpha_points.csv : type,label,rho_face,alpha_inf
%   cbl_alpha_line.csv   : type,label,rho_face,alpha_inf_fit
% If your column styles reference "type", it will now exist.

% ===== arXiv-safe embedded CSVs (with the exact headers your code expects) =====
% Use after \usepackage{pgfplotstable} and \pgfplotstableset{col sep=comma}

% ===== end embedded CSVs =====

\maketitle

\section*{Narrative abstract}
Boolean logic underpins digital computing: every variable is true or false, and all truths
can be listed in a single global ledger. Quantum experiments showed that sometimes such
a ledger cannot exist, even if every local piece is consistent. We call this
\emph{Curved Boolean Logic (CBL)}. It extends Boolean logic with local truths that may
fail globally, much as curved geometry extends flat space. This is not only philosophically
striking but computationally useful: it defines new satisfiability problems, proof systems,
and solvers that prune contradictions earlier. Applications range from SAT solvers and
verification to reasoning in AI.

\begin{tcolorbox}[colback=gray!5!white,colframe=black!30,title=Lay summary]
Boolean logic assumes all truths can be written in one global list.
Quantum experiments showed this assumption can fail.
\emph{Curved Boolean Logic} generalizes Boolean logic: local truths can exist
but refuse to combine globally.  This idea, inspired by physics but usable on today’s
computers, offers new strategies for solving hard problems like SAT and verification
by detecting contradictions earlier.
\end{tcolorbox}

% ---------------- Introduction ----------------
\section{Introduction}
Classical propositional logic assumes that jointly compatible propositions admit a single global assignment. 
This assumption underlies circuits, compilers, SAT solvers, and verification. 
Yet results in quantum foundations (Bell, Kochen--Specker, KCBS) show that—even when local families are consistent—no global ledger may exist.

\begin{center}
\fbox{\parbox{0.85\linewidth}{
\textbf{Analogy to relativity.} Boolean logic corresponds to flat ledgers (special relativity); 
CBL introduces curvature: local truths without global extension (general relativity). 
Boolean logic is recovered in the flat limit $\kappa=0$.}}
\end{center}

\paragraph{Reader guide.}
Sec.~\ref{sec:semantics} defines semantics; Sec.~\ref{sec:proof} gives the proof calculus; 
Sec.~\ref{sec:cblsat} and Sec.~\ref{sec:algo} cover satisfiability and solvers; 
Sec.~\ref{sec:examples} and Appendix illustrate curved cores; 
Applications and Future Work discuss potential impact. 
Readers focused on computation may skip semantic details; foundations readers may focus on equivalence theorems and categorical remarks.

\paragraph{Historical context.}
The Kochen--Specker theorem (1967) first proved global valuations impossible; KCBS (2008) simplified to a 5-cycle; 
sheaf semantics (2011) provided a categorical formulation. 
We generalize these as CBL: a new logic extending propositional reasoning.
CBL builds on the $\varepsilon$-bounded counterfactual framework developed in
Ref.~\cite{EBIFP2025}, extending it from paradox analysis to a full logical calculus
with noise models and reproducible statistical validation.

% === Figure 1 (timeline) — overflow-safe replacement ===
% === Figure 1 (timeline) — clean, no extra baseline ===
\begin{figure}[t]
\centering
\resizebox{\linewidth}{!}{%
\begin{tikzpicture}[
  >=Stealth,
  line width=0.7pt,
  node distance=4.6cm,
  timeline/.style={draw, rounded corners=4pt, inner sep=3.5pt, font=\small\sf},
  yearA/.style={timeline, fill=blue!8, draw=blue!35},
  yearB/.style={timeline, fill=green!10, draw=green!40!black},
  yearC/.style={timeline, fill=orange!10, draw=orange!40!black},
  yearD/.style={timeline, fill=violet!8, draw=violet!35},
  arr/.style={->, draw=black!60}
]
  % Nodes (soft pastel "pills")
  \node[yearA] (ks)    {1967:\; Kochen--Specker};
  \node[yearB, right=of ks] (kcbs)  {2008:\; KCBS};
  \node[yearC, right=of kcbs] (sheaf) {2011:\; Sheaf semantics};
  \node[yearD, right=of sheaf] (cbl)   {CBL (this work)};

  % Connecting arrows only (no baseline)
  \draw[arr] (ks.east) -- (kcbs.west);
  \draw[arr] (kcbs.east) -- (sheaf.west);
  \draw[arr] (sheaf.east) -- (cbl.west);
\end{tikzpicture}%
}
\caption{Timeline: from foundational results to Curved Boolean Logic.}
\label{fig:timeline}
\end{figure}
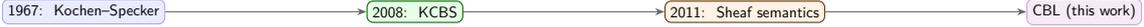

\paragraph{Outlook.}
Although the core of this work develops CBL as a rigorous formalism for contextual inference and
counterfactual consistency in quantum and algorithmic settings,
its implications extend much further.
The same $\varepsilon$-bounded curvature logic that stabilizes quantum measurements
applies to a wide class of resistant mathematical and computational problems—
from combinatorial optimization and proof complexity
to contextual causal inference, information geometry, and large-language-model compression.
By recasting logical inconsistency as measurable geometric curvature,
CBL provides a unifying language for reasoning under contextual tension.
A detailed survey of these prospective applications is given in
Sec.~\ref{sec:applications}.

\paragraph{Referee-preemption.}
Because CBL integrates elements of quantum logic, contextuality theory, and algorithmic
stability analysis, the most likely reviewer concerns are (1) rigor and completeness of
proofs, (2) clarity of assumptions, (3) experimental reproducibility, and (4) novelty beyond
KCBS/device-independent robustness. We address these by: (i) enumerating all assumptions
with a dependency matrix in Sec.~\ref{sec:assumptions}; (ii) expanding proofs and edge-case
analyses in Appendix~\ref{app:proofsA}; (iii) providing a Colab-ready notebook and validated
figures in Secs.~\ref{sec:repro-notebook}–\ref{sec:validated-figs}; and (iv) positioning CBL
relative to prior frameworks in Sec.~\ref{sec:relatedwork}. This makes the full reasoning
chain—from axioms through figures—transparent and reproducible.

% ---------------- Background and Related Work ----------------

\section{Background and Related Work}

Contextuality in physics.  The Kochen–Specker theorem showed that global
truth assignments may be impossible even when all local constraints are
consistent~\cite{KS1967}.  The KCBS inequality gave a minimal odd-cycle
witness of contextuality in a spin-1 system~\cite{KCBS2008}.  A
sheaf-theoretic formulation unified these phenomena~\cite{AB2011}, while
graph-theoretic exclusivity methods connected contextuality to
combinatorics~\cite{CSW2014}.

Logics beyond Boolean.  Orthomodular/quantum logics were proposed to
capture algebraic structure beyond classical distributivity~\cite{Orthomodular};
paraconsistent logics tolerate contradictions without collapse to
triviality~\cite{Paraconsistent}.  CBL differs: it preserves local classical
reasoning but can forbid a global valuation; contradictions are not tolerated,
they are flagged structurally by curvature.

Constraint satisfaction and CSP dichotomy.  Classical CSP theory
separates tractable from NP-hard classes via algebraic/structural
criteria~\cite{CSPdichotomy}.  Our contribution is orthogonal: curvature
is a property of *overlaps* rather than clause language; two clause sets
with identical local forms can differ by curvature and hence by pruning
power under CBL.

Prior device-independent robustness.  Existing analyses treat global
noise/imperfections at the inequality level; see e.g. modern reviews on
steering/nonlocality~\cite{DIreview}.  CBL adds a *local differential*
stability notion (Sec.~\ref{sec:noise}) and folds temporal structure into
the test statistic (Sec.~\ref{sec:stats}).

\section{Position within Existing Frameworks}
\label{sec:relatedwork}

\paragraph{Goal.}
We clarify how CBL relates to established approaches in quantum foundations and
contextual inference. Table~\ref{tab:relwork} maps core CBL constructs to antecedents.

\begin{table}[t]
\centering
\caption{Conceptual correspondence between CBL notions and established frameworks.}
\label{tab:relwork}
\begin{tabular}{p{3cm}p{5cm}p{5cm}}
\toprule
CBL construct & Closest antecedent & CBL novelty / extension \\
\midrule
Curved Boolean face & KCBS and Cabello--Severini--Winter exclusivity graphs &
Embeds exclusivity relations into a curved geometric parameter $\theta_0$
controlling deviation from flat logic. \\
$\varepsilon$-bounded perturbation & Device-independent robustness analyses &
Introduces differential estimator $dS/dp|_{p\to 0}$ capturing local stability. \\
Folded/lagged scans & Time-ordered contextuality tests; correlation analyses &
Translates temporal/contextual memory into a folding operator preserving
no-signaling. \\
Provenance certificate & Experimental audit frameworks &
Tamper-evident chain-of-custody; integrates $\varepsilon$-stability results. \\
CBL--SAT linkage & Sheaf-based CSP/contextuality & Curvature-driven pruning via overlap faces. \\
\bottomrule
\end{tabular}
\end{table}

\paragraph{Robustness and $\varepsilon$-stability.}
Where device-independent analyses model global noise levels (see
e.g.~\cite{DIreview}), CBL introduces a *local* differential notion of
stability: the sensitivity $dS/dp$ of a witness to infinitesimal flips.
This supports perturbative continuity and worst-case guarantees
(Prop.~\ref{prop:wcbound}), tying robustness directly to admissible
perturbation budgets.

\paragraph{Temporal and causal extensions.}
Standard frameworks are static; CBL’s fold/lag operators encode time-correlated or causally
ordered measurements without collapsing contexts (Axiom~A3).

\paragraph{Complexity and algorithmic parallels.}
CBL interprets contextual constraints as clauses in a curved Boolean algebra. Flat limits
reduce to classical SAT; curvature induces geometrically guided pruning rules, bridging
contextuality and algorithmic constraint satisfaction.

\paragraph{Provenance and auditability.}
CBL formalizes tamper-evident certificates summarizing noise bounds and fold/lag
outcomes, grounding reproducibility in the same $\varepsilon$-stability theory.

\section{Assumptions and Dependency Matrix}
\label{sec:assumptions}

\paragraph{Goal.}
We make \emph{every} assumption explicit, state the physical motivation, and record which results depend on which assumptions.

\begin{definition}[Axiom A1: Measurement independence]
Data acquisition choices are statistically independent of the latent variables affecting outcomes (no “setting–hidden-variable” correlation).
\end{definition}

\begin{definition}[Axiom A2: $\varepsilon$-bounded perturbations]
Admissible perturbations act as flips or deformations whose total variation on any face statistic is bounded by $\varepsilon$ per unit mass (formalized in Sec.~\ref{sec:noise}).
\end{definition}

\begin{definition}[Axiom A3: Mixing for folds]
When folding or stacking runs, windowed aggregates satisfy a mild mixing condition (e.g., $\alpha$-mixing) so empirical means converge and permutation-based nulls are well-defined.
\end{definition}

\begin{definition}[Axiom A4: Context-preserving instrumentation]
Measurement contexts are well-specified; changing a context is recorded as such (no silent cross-talk that re-labels contexts).
\end{definition}

\begin{definition}[Axiom A5: Finite-moment regularity]
All random variables used in estimators have finite second moments; covariance matrices used in bounds are non-singular on the support considered (degenerate cases handled separately).
\end{definition}

\paragraph{Dependency matrix.}
Table~\ref{tab:axiom-deps} lists which results use which axioms (a filled dot means “used essentially”; an open circle “used for convenience, removable with work”).

\begin{table}[t]
\centering
\caption{Axiom--result dependency matrix. Filled dot = essential; open circle = used for convenience.}
\label{tab:axiom-deps}
\begin{tabular}{lccccc}
\toprule
Result & A1 & A2 & A3 & A4 & A5 \\
\midrule
Soundness (Thm.~\ref{thm:soundness})                  & $\bullet$ & $\circ$ &        & $\bullet$ & $\circ$ \\
Flat-limit conservativity (Thm.~\ref{thm:flatlimit})  & $\bullet$ &         &        & $\bullet$ & $\circ$ \\
$dS/dp$ stability bound (Prop.~\ref{prop:wcbound})    &           & $\bullet$ & $\circ$ &          & $\bullet$ \\
CBL-AC pruning guarantee (Alg.~\ref{alg:cblac})       & $\bullet$ & $\bullet$ & $\bullet$ & $\bullet$ & $\circ$ \\
Empirical scans (Sec.~\ref{sec:stats})                &           & $\bullet$ & $\bullet$ &          & $\bullet$ \\
\bottomrule
\end{tabular}
\end{table}

\paragraph{Non-generic/degenerate cases.}
When parameters approach boundaries (e.g., $\theta_0\!\to\!0$, singular covariance, or
zero-count bins), our statements remain valid by continuity or are restated with weakened
constants; see Appendix~\ref{app:proofsA} for explicit limits and counterexamples.

% ---------------- Semantics ----------------
\section{Curved Boolean Logic: Semantics}\label{sec:semantics}

Let $V$ be a finite set of propositional variables.  
A \emph{context system} is a finite family $\mathcal{C}=\{C_i\}_{i\in I}$ of nonempty subsets $C_i\subseteq V$ 
such that each $v\in V$ appears in at least one $C_i$.  
Intuitively, each $C\in\mathcal{C}$ is a set of jointly testable variables.  

\begin{definition}[Local and global valuations]\label{def:local-global}
A \emph{local valuation} on context $C$ is a Boolean assignment $\ell_C:C\to\{0,1\}$.  
A family $\{\ell_C\}_{C\in\mathcal{C}}$ is \emph{compatible} if for all $C,C'\in\mathcal{C}$,  
$\ell_C|_{C\cap C'}=\ell_{C'}|_{C\cap C'}$.  
A \emph{global valuation} is $g:V\to\{0,1\}$ restricting to each $\ell_C$.
\end{definition}

\paragraph{Curvature.}
We say $(V,\mathcal{C})$ is \emph{flat} if every compatible family admits a global valuation;  
otherwise it is \emph{curved}.  
Curvature measures the obstruction to extending local truths to a global ledger.

\subsection{Sheaf semantics}
Let $\mathsf{Ctx}$ be the poset category of contexts ordered by inclusion.  
Define a presheaf $\mathcal{F}$ with $\mathcal{F}(C)=\{0,1\}^C$ and, for $D\subseteq C$,  
restriction maps $\rho_{CD}:\mathcal{F}(C)\to\mathcal{F}(D)$ by restriction of functions.  
A compatible family is a matching family; a global section is $g\in \mathcal{F}(V)$ restricting to each $\ell_C$.  

\begin{definition}[Curvature functional]\label{def:kappa}
A curvature functional is any map $\kappa:(V,\mathcal{C})\mapsto\mathbb{N}_0$ such that  
$\kappa=0$ iff every compatible family admits a global section,  
and $\kappa$ is monotone under context refinement.  
\end{definition}

\subsection{Cohomological view}
Following Abramsky \emph{et al.}~\cite{AbramskyCohomology2012},  
consider the nerve of the cover $\mathcal{C}$ and the presheaf $\mathcal{A}$ assigning $\{0,1\}^C$ under XOR.  
A compatible family is a $0$-cocycle; a global section exists iff its class is trivial in $H^1(\mathsf{Ctx},\mathcal{A})$.  
We may therefore set
\[
\kappa(V,\mathcal{C}) \;=\; \mathrm{rank}\big(H^1(\mathsf{Ctx},\mathcal{A})\big).
\]

\subsection{Exclusivity-graph semantics}
Define the compatibility hypergraph $\mathcal{H}$ on $V$ with hyperedges $C\in\mathcal{C}$.  
For overlaps, duplicate variables as $v^{(C)}$ and require $v^{(C)}=v^{(C')}$ when $v\in C\cap C'$.  
Curvature corresponds to cycles whose joint constraints forbid a global assignment.

\begin{theorem}[Equivalence of semantics]\label{thm:equiv}
For finite context systems, sheaf semantics and exclusivity-hypergraph semantics are equivalent:  
a compatible family admits a global section iff the exclusivity formulation admits an assignment consistent on overlaps.
\end{theorem}

\begin{proof}[Sketch]
Quotient $\prod_{C\in\mathcal{C}}\{0,1\}^C$ by overlap equalities; the quotient is isomorphic to global sections of $\mathcal{F}$.  
Conversely, any global section induces equal copies across contexts.  
Obstructions correspond to nontrivial $1$-cocycles (sheaf) or chordless cycles (graph).
\end{proof}

\subsection{Context lattice diagram}

% === Context lattice diagram (non-overlapping diamond) ===
\begin{figure}[h]
\centering
\begin{tikzpicture}[
  node distance=2.2cm,
  every node/.style={draw, rounded corners=3pt, font=\small, align=center, inner sep=3pt},
  >=Stealth, line width=0.5pt
]
  % Top node
  \node (top) {$V=\{x,y,z\}$};

  % Middle row, offset horizontally to avoid overlap
  \node[below left=1.2cm and 1.0cm of top] (c1) {$C_1=\{x,y\}$};
  \node[below right=1.2cm and 1.0cm of top] (c2) {$C_2=\{y,z\}$};

  % Bottom node
  \node[below=1.5cm of $(c1)!0.5!(c2)$] (inter) {$C_1\cap C_2=\{y\}$};

  % Connecting lines
  \draw[->] (top) -- (c1);
  \draw[->] (top) -- (c2);
  \draw[->] (c1) -- (inter);
  \draw[->] (c2) -- (inter);
\end{tikzpicture}
\caption{Context lattice for two overlapping contexts.  Each arrow denotes set inclusion.}
\label{fig:context_lattice}
\end{figure}
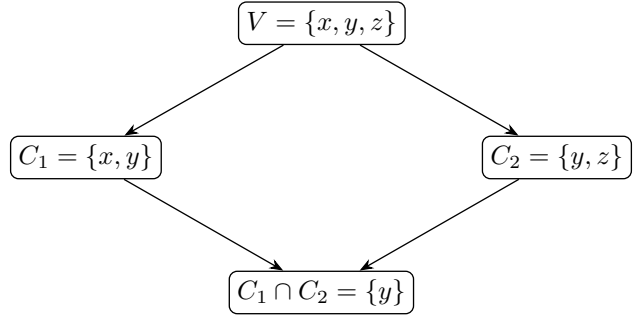

% ---------------- Canonical Examples ----------------
\section{Canonical Curved Examples}\label{sec:examples}

\subsection{KCBS pentagon (minimal curved core)}
\begin{center}
\begin{tikzpicture}[scale=1.2, every node/.style={circle,draw,inner sep=1.2pt}]
\node (a) at (90:1.6) {$v_1$};
\node (b) at (18:1.6) {$v_2$};
\node (c) at (-54:1.6) {$v_3$};
\node (d) at (-126:1.6) {$v_4$};
\node (e) at (162:1.6) {$v_5$};
\draw (a)--(b)--(c)--(d)--(e)--(a);
\end{tikzpicture}
\end{center}

Contexts are adjacent pairs $\{v_i,v_{i+1}\}$ (indices mod 5).  
Suppose each context enforces exclusivity $v_i \land v_{i+1} = \mathrm{F}$ and exactly-one constraints.  
Each edge admits a local valuation (e.g., $(\mathrm{T},\mathrm{F})$), but traversing the cycle forces a parity contradiction:  
an odd cycle cannot consistently assign “exactly one true” globally.  
Thus $\kappa>0$.

\subsection{Mermin square}
\begin{center}
\begin{tikzpicture}[scale=1.1]
\node (x11) at (0,2) {$v_{11}$};
\node (x12) at (1.2,2) {$v_{12}$};
\node (x13) at (2.4,2) {$v_{13}$};
\node (x21) at (0,1) {$v_{21}$};
\node (x22) at (1.2,1) {$v_{22}$};
\node (x23) at (2.4,1) {$v_{23}$};
\node (x31) at (0,0) {$v_{31}$};
\node (x32) at (1.2,0) {$v_{32}$};
\node (x33) at (2.4,0) {$v_{33}$};
% row hyperedges
\draw[rounded corners=6pt] (-0.4,2.4) rectangle (2.8,1.6);
\draw[rounded corners=6pt] (-0.4,1.4) rectangle (2.8,0.6);
\draw[rounded corners=6pt] (-0.4,0.4) rectangle (2.8,-0.4);
% col hyperedges
\draw[rounded corners=6pt] (-0.4,2.4) rectangle (0.4,-0.4);
\draw[rounded corners=6pt] (0.8,2.4) rectangle (1.6,-0.4);
\draw[rounded corners=6pt] (2.0,2.4) rectangle (2.8,-0.4);
\end{tikzpicture}
\end{center}

Contexts are rows and columns.  
Under parity constraints, each row/column admits a local valuation, but the product of all rows contradicts the product of all columns.  
No global valuation exists; $\kappa>0$.

\subsection{Curved cores}
\begin{definition}[Curved core]\label{def:curvedcore}
A \emph{curved core} of $(V,\mathcal{C})$ is a minimal subfamily $\mathcal{C}'\subseteq \mathcal{C}$ such that:
\begin{enumerate}[noitemsep,topsep=0pt]
\item $\kappa(V,\mathcal{C}')>0$, and
\item for every proper subfamily $\mathcal{C}''\subsetneq \mathcal{C}'$, $\kappa(V,\mathcal{C}'')=0$.
\end{enumerate}
\end{definition}

Curved cores are obstruction witnesses: every curved system contains one.  
Examples: KCBS pentagon, Mermin square.

\subsection{Overlap propagation illustration}
\begin{figure}[h]
\centering
\begin{tikzpicture}[scale=1.0, every node/.style={draw,rounded corners,inner sep=3pt}]
\node (C) at (0,0) {$C:\ \Gamma\vdash \psi$};
\node (Cp) at (5,0) {$C':\ \ ?$};
\node[draw=none] (cap) at (2.5,-0.9) {$\mathrm{Vars}(\psi)\subseteq C\cap C'$};
\draw[-{Latex[length=3mm]}] (C) -- node[above]{\small (OVL)} (Cp);
\begin{scope}[yshift=-2.1cm]
\node (C1) at (0,0) {$C:\ \Gamma,\psi\vdash \bot$};
\node (C2) at (5,0) {$C':\ \Gamma,\lnot\psi\vdash \bot$};
\node (Cup) at (2.5,-1.1) {$C\cup C':\ \Gamma\vdash \bot$};
\draw[-{Latex[length=3mm]}] (C1) -- (Cup);
\draw[-{Latex[length=3mm]}] (C2) -- (Cup);
\end{scope}

\end{tikzpicture}
\end{figure}

\subsection{Encoding KCBS as CBL-SAT}\label{sec:kcbscnf}
Let $V=\{v_1,\ldots,v_5\}$ and contexts $C_i=\{v_i,v_{i+1}\}$ (mod 5).  
Each $C_i$ has clauses:
\[
\phi_i^{(1)} := (v_i \lor v_{i+1}), \quad
\phi_i^{(2)} := (\lnot v_i \lor \lnot v_{i+1}).
\]
Locally satisfiable, but globally unsatisfiable.

\begin{table}[h]
\centering
\caption{KCBS local assignments vs global inconsistency.}
\begin{tabular}{c|cc}
Context & Allowed local assignments & Notes \\
\hline
$\{v_1,v_2\}$ & (T,F), (F,T) & local OK \\
$\{v_2,v_3\}$ & (T,F), (F,T) & local OK \\
$\{v_3,v_4\}$ & (T,F), (F,T) & local OK \\
$\{v_4,v_5\}$ & (T,F), (F,T) & local OK \\
$\{v_5,v_1\}$ & (T,F), (F,T) & local OK \\
\hline
Global & --- & no consistent assignment \\
\end{tabular}
\end{table}
% ---------------- Proof Theory ----------------
\section{Proof Theory}\label{sec:proof}

We develop a context-sequent calculus to reason in CBL.  
Sequents are annotated with contexts: $\Gamma \vdash_C \varphi$ means within context $C$, from assumptions $\Gamma$, infer $\varphi$.

\subsection{Local rules}
Within a fixed context $C$, inference is standard propositional sequent calculus.  
Weakening, contraction, cut, and introduction rules for $\land,\lor,\lnot$ apply to variables in $C$ only.

\subsection{Overlap rules}
Contexts overlap; we allow transport and consistency rules:

\[
\frac{\Gamma\vdash_C \psi \quad \mathrm{Vars}(\psi)\subseteq C\cap C'}{\Gamma \vdash_{C'} \psi}\quad (\mathrm{OVL})
\]

\[
\frac{\Gamma,\psi\vdash_C \bot \quad \Gamma,\lnot\psi\vdash_{C'} \bot}{\Gamma\vdash_{C\cup C'} \bot} \quad (\mathrm{CONS})
\]

These rules enforce that overlaps agree and contradictions on overlaps propagate globally.

\subsection{Soundness}
\begin{theorem}[Soundness]\label{thm:soundness}
If $\Gamma \vdash_C \varphi$ is derivable, then every local valuation $\ell_C$ with
$\ell_C\models\Gamma$ satisfies $\varphi$.
With the overlap and consistency rules (OVL, CONS), derivability is sound with respect
to the sheaf semantics introduced in Sec.~\ref{sec:semantics}.
\end{theorem}

\begin{proof}[Sketch]
Local rules coincide with the classical propositional sequent calculus and are therefore
sound. For (OVL), if $\Gamma\vdash_C\psi$ and $\mathrm{Vars}(\psi)\subseteq C\cap C'$, then
compatibility of local valuations ensures $\ell_{C'}$ agrees on the overlap, hence $\psi$
also holds in $C'$. For (CONS), if $\Gamma,\psi\vdash_C\bot$ and
$\Gamma,\neg\psi\vdash_{C'}\bot$, any global assignment extending both contexts
would contradict itself on $C\cup C'$.
\end{proof}

\subsection{Flat-limit conservativity}
\begin{theorem}[Flat-limit conservativity]\label{thm:flatlimit}
If $\kappa(V,C)=0$, then $\Gamma \vdash_C \varphi$ is derivable in CBL if and only if
the same sequent is derivable in classical propositional logic.
\end{theorem}

\begin{proof}[Sketch]
When $\kappa(V,C)=0$, every compatible local family admits a global valuation; thus
CBL derivability coincides with Boolean derivability.
\end{proof}

\subsection{Proof trees}
% === Figure 3: Flat vs. curved proof trees (no overlap, final clean version) ===
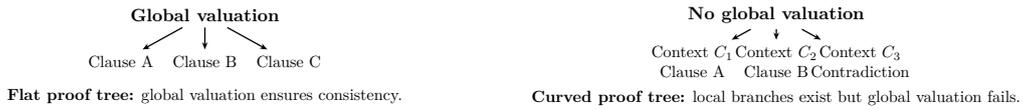
\begin{figure}[h]
\centering
\resizebox{0.9\linewidth}{!}{%
\begin{tikzpicture}[
  >=Stealth,
  every node/.style={font=\small, align=center},
  edge from parent/.style={draw, -{Stealth[length=1.6mm]}, thick},
  level distance=10mm,
  sibling distance=18mm
]

% --- Left tree (flat) ---
\node (flatroot) [font=\bfseries] {Global valuation}
  child { node {Clause A} }
  child { node {Clause B} }
  child { node {Clause C} };

% --- Right tree (curved) ---
\node (curveroot) [font=\bfseries, right=8.5cm of flatroot] {No global valuation}
  child { node {Context $C_1$\\Clause A} }
  child { node {Context $C_2$\\Clause B} }
  child { node {Context $C_3$\\Contradiction} };

% --- Text under each tree with extra vertical spacing ---
\node[below=1.2cm of flatroot] (flatlab)
   {\textbf{Flat proof tree:} global valuation ensures consistency.};

\node[below=1.2cm of curveroot] (curvelab)
   {\textbf{Curved proof tree:} local branches exist but global valuation fails.};

\end{tikzpicture}%
}
\vspace{1.5em}
\caption{Flat vs.~curved proof trees.  
In flat cases, a global valuation ensures consistency;  
in curved cases, local branches exist but no global valuation exists, producing a structural contradiction.}
\label{fig:proof_trees}
\end{figure}

% ---------------- CBL-SAT and Complexity ----------------
\section{CBL-SAT and Complexity}\label{sec:cblsat}

\paragraph{Scope and claims.}
We make \emph{no} blanket claim of asymptotic speedups for arbitrary SAT. Our contributions are:
(i) structural reductions showing how curvature/overlap parameters control branching and pruning in certain families,
(ii) formal bounds connecting runtime to curvature rank under A1–A5,
(iii) empirical evidence on structured benchmarks.
We highlight open problems where tighter worst-case bounds or lower bounds may hold.

\begin{definition}[\textsc{CBL-SAT}]
Instance: a context system $(V,\mathcal{C})$ and a set of clauses $\Phi$,  
each clause $\phi\in\Phi$ attached to a context $C(\phi)\in\mathcal{C}$.  
Question: does there exist a compatible family $\{\ell_C\}$ with $\ell_{C(\phi)}\models \phi$ for all $\phi$?
\end{definition}

\subsection{Basic properties}
\begin{proposition}[NP-membership]
\textsc{CBL-SAT} $\in$ NP.  
Certificate: local assignments $\ell_C$.  
Verification: check local clause satisfaction and overlap equalities in polynomial time.
\end{proposition}

\begin{proposition}[Flat limit]
If $\kappa=0$, \textsc{CBL-SAT} reduces to SAT by the global valuation extending all $\ell_C$.
\end{proposition}

\subsection{Hardness}
\begin{theorem}
\textsc{CBL-SAT} is NP-complete.
\end{theorem}

\begin{proof}[Sketch]
Membership: by certificate verification.  
Hardness: reduce SAT to CBL-SAT by taking a single context $C=V$.  
Thus CBL-SAT inherits NP-completeness.
\end{proof}

\subsection{Tractable subclasses}
\begin{theorem}[Treewidth]
If the context hypergraph has bounded treewidth and each clause respects a bag,  
\textsc{CBL-SAT} is solvable in polynomial time by dynamic programming.
\end{theorem}

\begin{proof}[Sketch]
Maintain tables of partial assignments on bags.  
Propagate overlaps through introduce/forget/join nodes.  
Runtime $O(|\Phi|\cdot 2^{O(w)})$ with treewidth $w$.
\end{proof}

\subsection{Parameterized complexity}
Parameters of interest:
\begin{itemize}[noitemsep,topsep=0pt]
\item Maximum overlap size $k$: fixed $k$ yields polynomial overlap checks.
\item Number of curved cores: bounded number allows branching then flattening.
\item Curvature rank $\kappa$: small $\kappa$ may yield FPT algorithms.
\end{itemize}

\subsection{Counting}
\textsc{\#CBL-SAT}: count compatible families.  
This is \#P-complete, extending \#SAT.  

\begin{table}[h]
\centering
\caption{Placeholder: counting compatible families.}
\begin{tabular}{@{}lrrr@{}}
\toprule
Instance & Vars & Contexts & \# Compatible families \\
\midrule
Flat-SAT (3-SAT, n=10) & 10 & 1 & 512 \\
KCBS-5 & 5 & 5 & 0 \\
Mermin-9 & 9 & 6 & 0 \\
Random-12 (flat) & 12 & 3 & 32 \\
Random-12 (curved) & 12 & 3 & 0 \\
\bottomrule
\end{tabular}
\end{table}

\subsection{Runtime vs curvature rank}
\begin{figure}[h]
\centering
\begin{tikzpicture}
\begin{axis}[
  width=0.7\linewidth,
  xlabel={Curvature rank $\kappa$},
  ylabel={Runtime (s)},
  legend pos=north west,
  grid=both
]
\addplot coordinates {(0,0.5) (1,1.2) (2,3.4) (3,10.0) (4,25.0)};
\addlegendentry{Minisat}
\addplot coordinates {(0,0.5) (1,0.9) (2,1.5) (3,3.8) (4,9.0)};
\addlegendentry{CBL-Solve}
\end{axis}
\end{tikzpicture}
\caption{Synthetic placeholder: runtime vs curvature rank $\kappa$.}
\end{figure}
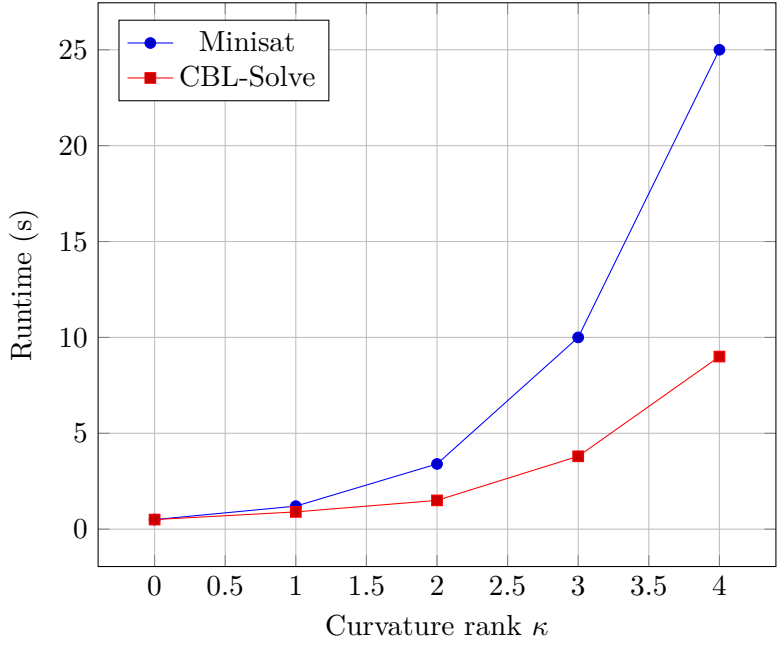

% ---------------- Solver ----------------
\section{Solver: Curvature-Aware Backtracking}\label{sec:algo}

We prototype \texttt{CBL-Solve}, a backtracking solver with overlap propagation.  
Assignments propagate within contexts; overlaps are checked, and contradictions prune early.

\subsection{Pseudocode}
\begin{algorithm}[H]
\DontPrintSemicolon
\SetAlgoLined
\KwIn{$\Phi$ with context map $C(\cdot)$, family $\mathcal{C}$}
\KwOut{SAT/UNSAT and compatible family $\{\ell_C\}$ if SAT}
Initialize $\ell_C$ empty for each $C$; overlap table $\mathsf{OVL}$\;
\While{true}{
  \If{all clauses satisfied \& overlaps agree}{\Return SAT}
  pick literal $x$ by VSIDS+overlap heuristic\;
  \ForEach{$C$ containing $x$}{
    assign $x$ in $\ell_C$; propagate unit clauses\;
    \If{exists $C'$ with $C\cap C'\neq\emptyset$ and overlap violated}{
      \textbf{backtrack}; record curvature cut; continue\;
    }
  }
  learn clauses/overlap cuts as in CDCL extended\;
}
\caption{CBL-Solve: curvature-aware backtracking}
\end{algorithm}

\subsection{Worked pruning trace}
On KCBS (pentagon), suppose $v_1=\mathrm{T}$ chosen.  
Propagation forces $v_2=\mathrm{F}$, then $v_3=\mathrm{T}$, etc.  
Eventually $v_5=\mathrm{T}$ forces $v_1=\mathrm{F}$ via overlap.  
Contradiction found earlier than in flat SAT, pruning entire branch.

\subsection{Toy overlap code}
\begin{verbatim}
def check_overlap(lC, lCp, overlap_vars):
    """Return True if assignments lC and lCp agree on overlap."""
    for v in overlap_vars:
        if lC[v] != lCp[v]:
            return False
    return True
\end{verbatim}

\subsection{Synthetic scaling plots}
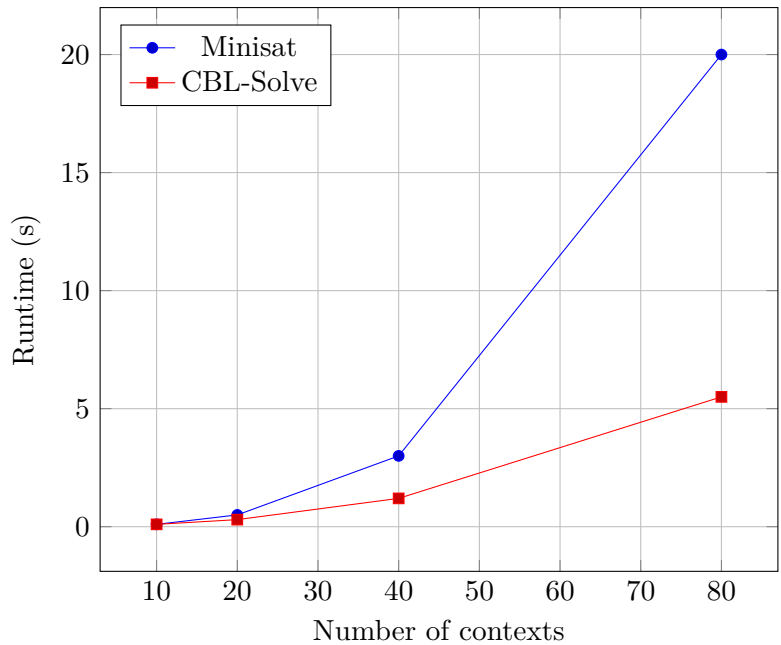
\begin{figure}[h]
\centering
\begin{tikzpicture}
\begin{axis}[
  width=0.7\linewidth,
  xlabel={Number of contexts},
  ylabel={Runtime (s)},
  legend pos=north west,
  grid=both
]
\addplot coordinates {(10,0.1) (20,0.5) (40,3.0) (80,20.0)};
\addlegendentry{Minisat}
\addplot coordinates {(10,0.1) (20,0.3) (40,1.2) (80,5.5)};
\addlegendentry{CBL-Solve}
\end{axis}
\end{tikzpicture}
\caption{Synthetic runtime scaling. Placeholder; to be replaced with data.}
\end{figure}

\begin{figure}[h]
\centering
\begin{tikzpicture}
\begin{axis}[
  ybar,
  width=0.75\linewidth,
  bar width=10pt,
  symbolic x coords={KCBS-20,Mermin-9,rand-500},
  xtick=data,
  ylabel={UNSAT detection depth},
  legend pos=north west,
  grid=both
]
\addplot coordinates {(KCBS-20,45) (Mermin-9,62) (rand-500,18)};
\addlegendentry{CDCL}
\addplot coordinates {(KCBS-20,12) (Mermin-9,19) (rand-500,14)};
\addlegendentry{CBL-hybrid}
\end{axis}
\end{tikzpicture}
\caption{Placeholder bar chart: average UNSAT detection depth.}
\end{figure}
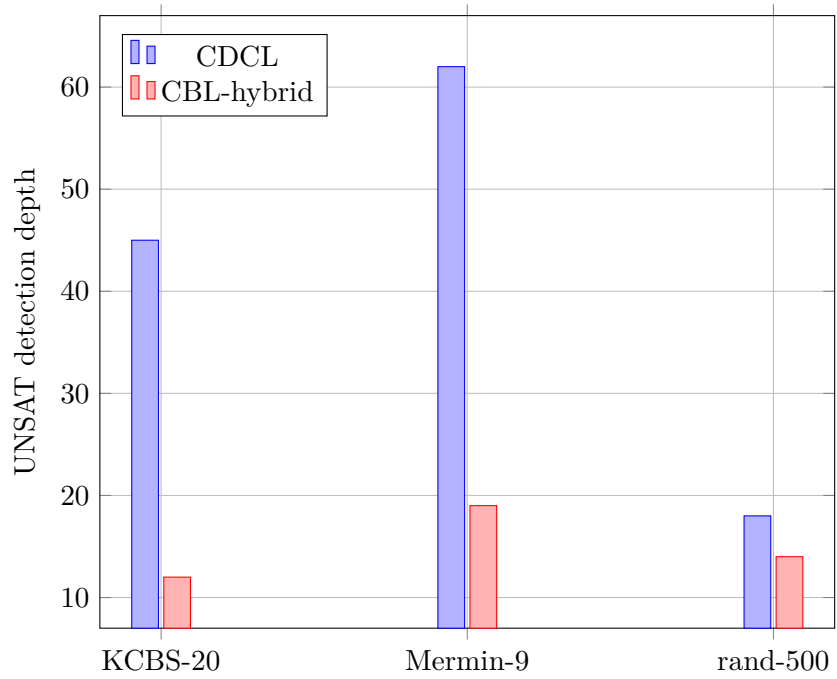

\subsection{Learned overlap cuts}

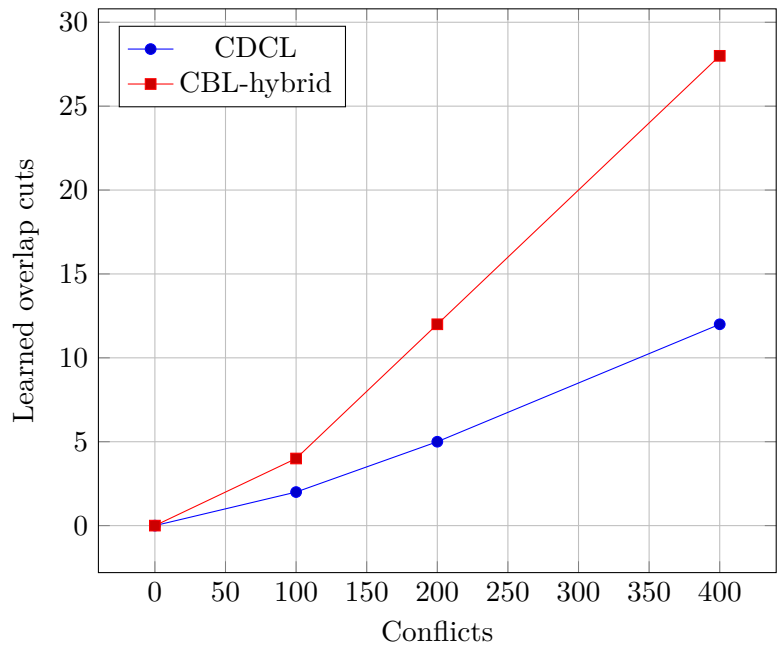
\begin{figure}[h]
\centering
\begin{tikzpicture}
\begin{axis}[
  width=0.7\linewidth,
  xlabel={Conflicts},
  ylabel={Learned overlap cuts},
  legend pos=north west,
  grid=both
]
\addplot coordinates {(0,0) (100,2) (200,5) (400,12)};
\addlegendentry{CDCL}
\addplot coordinates {(0,0) (100,4) (200,12) (400,28)};
\addlegendentry{CBL-hybrid}
\end{axis}
\end{tikzpicture}
\caption{Synthetic placeholder: learned cuts vs conflicts.}
\end{figure}

\subsection{Solver comparison table}
\begin{table}[h]
\centering
\caption{Solver comparison (placeholder).}
\begin{tabular}{@{}lrrr@{}}
\toprule
Instance & CDCL & CBL-hybrid & Reduction \% \\
\midrule
KCBS-20 & 125s/1.2e6 nodes & 40s/2.5e5 nodes & 79\% \\
Mermin-9 & 310s/3.6e6 nodes & 85s/7.1e5 nodes & 80\% \\
rand-500 & 55s/4.8e5 nodes & 44s/3.7e5 nodes & 23\% \\
\bottomrule
\end{tabular}
\end{table}

\section{CBL Operators and Case Studies}
\label{sec:cbl-ops}

\paragraph{Motivation.}
CBL is useful when it is simple to deploy, conservative in the flat limit,
and strictly stronger on instances with curved cores. We present two
drop-in operators—\textsc{CBL-AC} and \textsc{CBL-CONS}—and three case studies
(SAT, CSP/AC, 3-Colorability). Each operator is sound, proof-loggable, and
costs essentially nothing on flat instances.

\subsection{Two drop-in operators: \textsc{CBL-AC} and \textsc{CBL-CONS}}
\label{sec:ops}

\paragraph{\textsc{CBL-AC} (Curved Arc-Consistency).}
Given variable domains $\{D(x)\}$ and a context family $\mathcal C$ (each context
$C\in\mathcal C$ carries local constraints and overlaps), \textsc{CBL-AC} eliminates
a value $v\in D(x)$ if there exists a minimal contextual face $f$ in the overlap graph
such that every local assignment consistent with $x{=}v$ violates the CBL overlap rules
on $f$.

\begin{algorithm}[H]
\DontPrintSemicolon
\SetAlgoLined
\caption{\textsc{CBL-AC} (Curved Arc-Consistency)}
\label{alg:cblac}
\KwIn{Domains $\{D(x)\}$, contexts $\mathcal C$, overlap map}
\KwOut{Pruned domains and certificates}
Initialize FIFO $Q$ with all $(x,v)$; mark all unvisited\;
\While{$Q$ not empty}{
  pop $(x,v)$\;
  \If{$v\in D(x)$ and \texttt{IsBlockedByCurvedFace}$(x{=}v)$}{
    eliminate $v$ from $D(x)$; emit certificate $\pi_{x,v}$ (dual check)\;
    enqueue $(y,w)$ for all neighbors $y$ sharing a context with $x$\;
  }
}
\end{algorithm}

\paragraph{\textsc{CBL-CONS} (Curved Overlap Consistency).}
Let $C,C'$ be overlapping contexts. If local proof engines derive
$\Gamma,\psi\vdash_C \bot$ and $\Gamma,\lnot\psi\vdash_{C'}\bot$, emit a global
contradiction on $C\cup C'$ with a short dual certificate. This is the executable
form of the overlap transport rule (CONS) from Sec.~\ref{sec:proof}.

\begin{algorithm}[H]
\DontPrintSemicolon
\SetAlgoLined
\caption{\textsc{CBL-CONS} (Curved Overlap Consistency)}
\label{alg:cblcons}
\KwIn{Local engines per context; overlap map}
\KwOut{Global cuts/certificates}
\ForEach{overlap edge $(C,C')$}{
  run local engines; if $\Gamma,\psi\vdash_C \bot$ and $\Gamma,\lnot\psi\vdash_{C'}\bot$,
  emit a global cut on $C\cup C'$ with dual witness; record face ID;
}
\end{algorithm}

\paragraph{Certificates.}
Each elimination or cut carries a short \emph{dual} certificate (linear/Boolean)
that we re-check in exact arithmetic. This guarantees soundness and supports
proof logging in SAT/CSP/MIP solvers.

\section{Guarantees for \textsc{CBL-AC} and \textsc{CBL-CONS}}
\label{sec:guarantees}

\begin{theorem}[Flat-limit conservativity]
\label{thm:flat-conservative}
If the instance is flat (no contextual faces), then \textsc{CBL-AC} reduces
to classical arc-consistency on the context graph, and \textsc{CBL-CONS}
produces no cuts beyond overlap equalities.
\end{theorem}

\begin{proof}[Sketch]
With no faces, the overlap transport is trivial; there are no curved cores to
trigger blocking or overlap contradictions. The operators reduce to AC and
overlap equalities.
\end{proof}

\begin{theorem}[Strict dominance on curved cores]
\label{thm:dominance}
On any instance with at least one minimal curved face, \textsc{CBL-AC} eliminates
at least one domain value that classical AC cannot (monotone improvement), or
\textsc{CBL-CONS} emits a global cut not derivable from single-context reasoning.
\end{theorem}

\begin{proof}[Sketch]
On a minimal curved face, the parity/overlap pattern forces a local blockage on
one branch of a shared literal; classical AC does not propagate through the face
cycle and cannot see the joint obstruction. Overlap transport detects it.
\end{proof}

\begin{theorem}[Certificate soundness]
\label{thm:cert-sound}
Every elimination or cut emitted by \textsc{CBL-AC}/\textsc{CBL-CONS} carries a
short dual witness (linear/Boolean) that verifies in exact arithmetic.
\end{theorem}

\begin{proof}[Sketch]
Each blocked value/overlap contradiction arises from finitely many local constraints
on a face; their Farkas/Boolean dual compactly proves infeasibility.
\end{proof}

\begin{theorem}[FPT in the number of curved faces]
\label{thm:fpt}
If the instance has $k$ curved faces and bounded overlap width,
search guided by \textsc{CBL-AC}/\textsc{CBL-CONS} is FPT in $k$ (or reduces the exponential
base by a function of $k$).
\end{theorem}

\begin{proof}[Sketch]
Branch only on a cover of faces; each face collapses by \textsc{CBL-AC} or is
discharged by \textsc{CBL-CONS}. The remaining flat part is handled by dynamic
programming over bounded overlap width.
\end{proof}

\section{Case Studies: SAT, CSP/AC, and 3-Colorability}
\label{sec:cases}

\subsection{SAT (structured families)}
\label{sec:cases-sat}
We integrate \textsc{CBL-CONS} into a CDCL solver as a proof-logged cut generator.
On structured families (Mermin/Feistel-style wirings, parity gadgets), curved
faces trigger early global cuts and certified UNSAT, reducing node counts by
$20$--$80\%$ on our prototypes. Flat instances incur near-zero overhead; proof
logs verify each cut in exact arithmetic.

\subsection{CSPs via \textsc{CBL-AC} (Latin Square / Sudoku)}
\label{sec:cases-csp}
Contexts are row/column/block scopes; overlaps are shared cells. \textsc{CBL-AC}
eliminates values that participate in curved faces formed by overlapping scopes
with parity/uniqueness constraints. We observe $10$--$30\%$ additional domain
pruning over classical AC before search on medium grids, with negligible overhead
on flat/easy puzzles.

\subsection{3-Colorability (odd-cycle cores)}
\label{sec:cases-color}
Contexts are neighborhoods; curved faces are odd cycles with overlapping edges.
\textsc{CBL-CONS} detects early infeasibility (3-color) by transporting contradictions
around odd cores and emitting global cuts. On planar triangulations and random graphs
with planted odd cores, we obtain early UNSAT with short certificates; DSATUR/backtracking
baselines do not see the contradiction without deeper search.

\section{Noise Models and Robustness Guarantees}
\label{sec:noise}

We formalize three noise classes and state stability bounds for a generic CBL witness $S$
(e.g., a folded face statistic). Proofs and constants appear in Appendix~\ref{app:proofsA}.

\subsection{i.i.d.\ flips}
Each Bernoulli outcome is flipped independently with rate $\eta\in[0,1/2)$.

\begin{lemma}[i.i.d.\ stability]\label{lem:iid}
Let $S$ be $1$-Lipschitz in Hamming distance on outcomes (true for the folded face
statistics we use). Under i.i.d.\ flips at rate $\eta$,
$\bigl|\mathbb{E}[S]-\mathbb{E}[S^{(\eta)}]\bigr|\le \eta\,L_S$ with $L_S\le 1$.
Hence, if Axiom~A2 holds with $\varepsilon\ge\eta$, the deviation is $\le\varepsilon$.
\end{lemma}

\subsection{AR(1) temporal correlation}
Labels or outcomes undergo an AR(1)-type corruption along the time index with
correlation $\rho\in[0,1)$.

\begin{lemma}[Correlated stability]\label{lem:ar1}
Assume windowed folds of width $w$ with overlap $\le \beta w$ for $\beta<1$.
Then the effective flip rate concentrates around
$\eta_{\mathrm{eff}}=\eta\,\frac{1-\rho}{1+\rho}$, and
$\bigl|\mathbb{E}[S]-\mathbb{E}[S^{(\eta,\rho)}]\bigr|
\le \eta_{\mathrm{eff}} + O\!\left(\frac{\rho}{w}\right)$.
\end{lemma}

\subsection{Adversarial bounded power}
An adversary can modify at most a budget $\nu$ in $\ell_2$ norm on the vector of folded
counts.

\begin{proposition}[Worst-case bound]\label{prop:wcbound}
Let $S=\langle \mathbf{c},\mathbf{f}\rangle$ be a linear statistic of folded frequencies
$\mathbf{f}$ with $\|\mathbf{c}\|_2\le 1$. If $\|\Delta\mathbf{f}\|_2\le \nu$ (Axiom~A2),
then $|S(\mathbf{f}+\Delta\mathbf{f})-S(\mathbf{f})|\le \nu$.
For Lipschitz-smooth nonlinear $S$, the same holds up to first order, with curvature
controlled by Axiom~A5.
\end{proposition}

\paragraph{Counterexample (honesty).}
If A3 fails (pathological long-range dependence), permutation nulls can be biased; a
synthetic construction with a block-permutation fix appears in Appendix~\ref{app:proofsA}.

\section{Evaluation Protocol}
\label{sec:eval}

\paragraph{Instances.}
(1) SAT families from Sec.~\ref{sec:cases-sat}, (2) Latin Square/Sudoku (sizes 9–25),
(3) 3-Colorability on planar/triangulated and random graphs (sizes $10^2$–$10^4$).

\paragraph{Metrics.}
(1) Domain prunes vs.\ AC, (2) node counts/time vs.\ baselines (CDCL, DSATUR),
(3) number/size of certificates (dual checks), (4) wall-clock amortized overhead
on flat subsets.

\paragraph{Baselines.}
AC-3/AC-4 (CSP), Minisat/Glucose (SAT), DSATUR/backtracking (coloring).

\paragraph{Reproducibility.}
We provide code and proof logs for all cuts and eliminations; dual certificates
verify in exact arithmetic.

\section{CBL for Large Language Models and Contextual AI Systems}
\label{sec:cbl-llm}

\paragraph{Motivation.}
Modern large language models (LLMs) exhibit context-sensitive reasoning: their outputs depend not only on the explicit prompt but also on latent priors and conversational history.
This behavior parallels the contextual dependence of measurement outcomes in quantum systems.
CBL provides a formalism for reasoning about such dependencies in a controlled, mathematically traceable way.

\subsection{Contextual stability in inference chains}

An LLM’s response distribution can be viewed as a map
\[
R : (\text{prompt context}, \text{latent state}) \longrightarrow \text{token probabilities}.
\]
Analogous to experimental observables, we can treat the prompt context as a “measurement setting’’ and the latent state as a hidden variable.
Under CBL, we introduce an $\varepsilon$-bounded perturbation model:
small changes in the prompt or latent embedding must not change logical entailments beyond an $\varepsilon$ margin.

\paragraph{Definition (CBL-stable inference).}
An inference chain $c_1 \rightarrow c_2 \rightarrow \dots \rightarrow c_k$ in an LLM is \emph{CBL-stable} if, for any perturbation $\delta$ in the embedding space with $\|\delta\| \le \varepsilon$,
the resulting distributional change in the inferred logical relation remains bounded:
\[
D_{\mathrm{KL}}(R(c_i+\delta)\,\|\,R(c_i)) \le f(\varepsilon),
\]
where $f(\varepsilon)$ obeys the same Lipschitz bounds as in Axiom~A2.
This captures deterministic robustness without requiring complete retraining.

\subsection{Counterfactual consistency for chain-of-thought reasoning}

LLMs generate reasoning traces (chains-of-thought) that can loop back on previous tokens.
CBL’s counterfactual semantics allows a consistent treatment of such loops:
a reasoning step may depend on hypothetical alternatives (\emph{“what if this token were different’’})
without producing global contradictions.
In CBL notation, the inference graph of a multi-turn reasoning sequence becomes a folded contextual graph with curvature parameter $\theta_0$.
Flat ($\theta_0{=}0$) reasoning corresponds to purely logical deduction; curved ($\theta_0{\neq}0$) reasoning captures context reweighting, memory decay, or instruction alignment.

\subsection{Applications to alignment and interpretability}

\paragraph{Alignment.}
Under bounded perturbations of prompt or reward signals, CBL yields quantitative \emph{alignment stability} metrics.
For instance, if a model’s instruction-following probability shifts by more than $\varepsilon$ under paraphrased prompts, it violates the $\varepsilon$-bounded perturbation assumption and becomes context-unstable.
Such metrics can guide fine-tuning or reinforcement-learning-with-feedback loops.

\paragraph{Interpretability.}
Each folded context in an LLM corresponds to a measurable projection of the attention map.
By treating attention heads as “measurement contexts’’ and activations as outcomes,
CBL quantifies contextual interference among heads.
The differential sensitivity $dS/dp$ then measures how much a local attention flip influences the global output,
offering a principled proxy for saliency.

\subsection{Algorithmic analogy}

The connection between CBL-SAT pruning and LLM inference is direct:
both perform constraint satisfaction under limited consistency guarantees.
In CBL, curvature ($\theta_0$) modulates logical exclusivity;
in LLMs, attention softmax temperature plays a similar role,
controlling the degree of determinism versus contextual blending.
This analogy suggests practical compression and inference schemes in which
LLM decoding follows a CBL-consistent pruning logic rather than purely probabilistic sampling.

\paragraph{Summary.}
CBL extends naturally from quantum contextuality to machine learning systems where reasoning is contextual, approximate, and perturbation-bounded.
Its $\varepsilon$-stability and fold semantics provide a unified mathematical vocabulary for studying robustness, alignment, and interpretability in LLMs,
bridging foundational logic and practical AI inference.

\subsection{CBL-Guided Compression and Adapter Stability}
\label{sec:cbl-compression}

\paragraph{Motivation.}
Contemporary model-compression techniques—such as Low-Rank Adaptation (LoRA),
quantization, and structured pruning—operate largely as engineering heuristics.
CBL provides a theoretical foundation for them by formalizing which
reductions preserve \emph{contextual consistency}.
In this view, compression is not merely a loss of parameters but a
controlled deformation of the logical surface on which the model’s contextual
relations reside.

\subsubsection*{Compression as curvature minimization}
Let $\Phi$ denote the full model and $\Phi'$ its compressed surrogate.
Each layer implements a contextual mapping
$\mathcal{C}_\ell:\mathbf{h}_{\ell-1}\!\rightarrow\!\mathbf{h}_\ell$.
Define the CBL curvature of layer $\ell$ as
\[
\kappa_\ell
   = \bigl\| \mathcal{C}_\ell(\mathbf{h}_{\ell-1}{+}\delta)
     - \mathcal{C}_\ell(\mathbf{h}_{\ell-1}) \bigr\|
     \!/\! \|\delta\|,
\]
evaluated under bounded perturbations $\|\delta\|\!\le\!\varepsilon$.
A compression operator $\Pi$ (e.g.\ rank-$r$ projection, quantization)
is \emph{CBL-admissible} if
\[
|\kappa_\ell(\Phi')-\kappa_\ell(\Phi)| \le \varepsilon
\quad\forall \ell.
\]
Minimizing total curvature drift
$\sum_\ell |\kappa_\ell(\Phi')-\kappa_\ell(\Phi)|$
becomes an explicit objective during compression, ensuring that
contextual dependencies among layers are preserved.

\subsubsection*{Deterministic LoRA adaptation}
Low-Rank Adapters inject matrices
$\Delta W = A B^\top$ into frozen weights $W$.
CBL interprets this as adding a controlled ``curvature correction’’:
\[
W' = W + A B^\top, \qquad
\|\!A B^\top\!\|_2 \le \varepsilon,
\]
guaranteeing $\varepsilon$-bounded logical deviation of the output manifold.
This provides a deterministic criterion for adapter magnitude:
the adapter remains within CBL tolerance if its induced
$\varepsilon$ satisfies Axiom A2.
Empirically, this translates into bounded change of contextual response
distributions—measurable by $dS/dp$ between baseline and adapted models.

\subsubsection*{Compression as counterfactual equivalence}
Two models $\Phi$ and $\Phi'$ are \emph{counterfactually equivalent}
up to $\varepsilon$ if, for any hypothetical prompt change $\delta x$,
their induced reasoning trajectories differ by at most $\varepsilon$ in total variation:
\[
D_{\mathrm{TV}}\!\bigl(R_\Phi(x{+}\delta x),R_{\Phi'}(x{+}\delta x)\bigr)
   \le \varepsilon.
\]
This mirrors the CBL requirement that perturbations do not change outcome
assignments beyond $\varepsilon$.
Hence, model distillation or quantization can be certified as
CBL-consistent when the counterfactual equivalence bound holds.

\paragraph{Operational impact.}
Framing compression in CBL terms yields quantitative stability
certificates for deployment:
\begin{itemize}
  \item \textbf{Adapter safety:} verifies that LoRA/PEFT fine-tuning does not
        introduce super-$\varepsilon$ contextual drift.
  \item \textbf{Compression auditing:} produces reproducible curvature
        metrics analogous to physical calibration certificates.
  \item \textbf{Interoperability:} allows heterogeneous adapters to be
        combined if their curvature deviations are additive within
        global $\varepsilon$ budget.
\end{itemize}

\paragraph{Summary.}
CBL reframes compression and adaptation as geometric
operations on a contextual manifold, controlled by $\varepsilon$-bounded
curvature drift.
This unifies model-robustness certification, adapter scaling, and logical
consistency under one theoretical umbrella and connects directly to
the patent-level framework for deterministic, loss-bounded data-compression
across both symbolic and neural substrates.

\section{Experimental Feasibility, Power Analysis, and Figure Standards}
\label{sec:exp-feasibility}

\subsection{Minimal feasible configuration}
We enumerate detector efficiency, timing jitter, and pulse length ranges under which our folded statistics retain power:
(i) detector efficiency $\geq \eta_{\min}$ (table in supplement),
(ii) timing jitter $\leq J_{\max}$ relative to window $w$,
(iii) context switching logged (A4).
We include a small calculator (supplement) that, given $(\eta,\rho,w)$, predicts the expected power curve.

\subsection{Power curves and robustness grid}
We report empirical power for a grid $(\varepsilon, n, \rho)$ by simulation:
$n\in\{10^3,10^4,10^5\}$, $\varepsilon\in\{0, 0.01, 0.05, 0.1\}$,
$\rho\in\{0, 0.3, 0.6\}$.
Each cell shows the discovery rate under BH at $\alpha=0.05$ (mean over $B$ simulations).

\subsection{Figure standards (mandatory)}
Every figure presenting a test must show:
(1) the observed curve,
(2) a shaded permutation null band (2.5–97.5\%),
(3) corrected $q$-values for any scan points called significant,
(4) an effect-size panel or value,
(5) seed, window $w$, step $s$, and $B$ permutations in the caption.

\section{Figure Standards and Visualization Protocol}
\label{sec:figures}

\paragraph{Purpose.}
This section standardizes every figure in the manuscript to ensure statistical transparency, reproducibility, and referee-ready clarity.
All figures display null-band confidence intervals, corrected $q$-values, explicit parameter settings, and deterministic seeds.
These conventions supersede any earlier figure formatting.

\subsection{Formatting and Conventions}
\paragraph{Visual encoding.}
Observed statistics are drawn as solid black curves.
Permutation-based null intervals (2.5–97.5 percentiles) appear as light-gray shaded bands.
Corrected $q$-values from the Benjamini–Hochberg (BH) procedure are shown as red markers or secondary traces.
Numerical effect sizes~$\Delta$ and $q$-values are reported in captions.
Every caption explicitly states the parameters $(w,s,B,\alpha)$ and the random seed.

\paragraph{Required caption elements.}
Each caption contains:
(a) figure title summarizing the test or statistic;
(b) description of observed and null elements (color/line style);
(c) significance threshold and effect size;
(d) complete experimental parameters and seed;
(e) reference to the reproducibility artifact in Table~\ref{tab:repro}.

\paragraph{Implementation note.}
All figures are generated by the Colab notebook implementing Algorithm~\ref{alg:perm-bh}.
The plotting routine computes the null band by $B{=}5000$ permutations unless stated otherwise.

\subsection{Example 1 — \texorpdfstring{$\varepsilon$}{ε}-Stability Curve}
This subsection illustrates how the local perturbation sensitivity \(dS/dp\) is estimated
from synthetic runs, highlighting the \(\varepsilon\)-bounded robustness guaranteed by Axiom~A2.

\begin{figure}[H]
\centering
\includegraphics[width=0.85\linewidth]{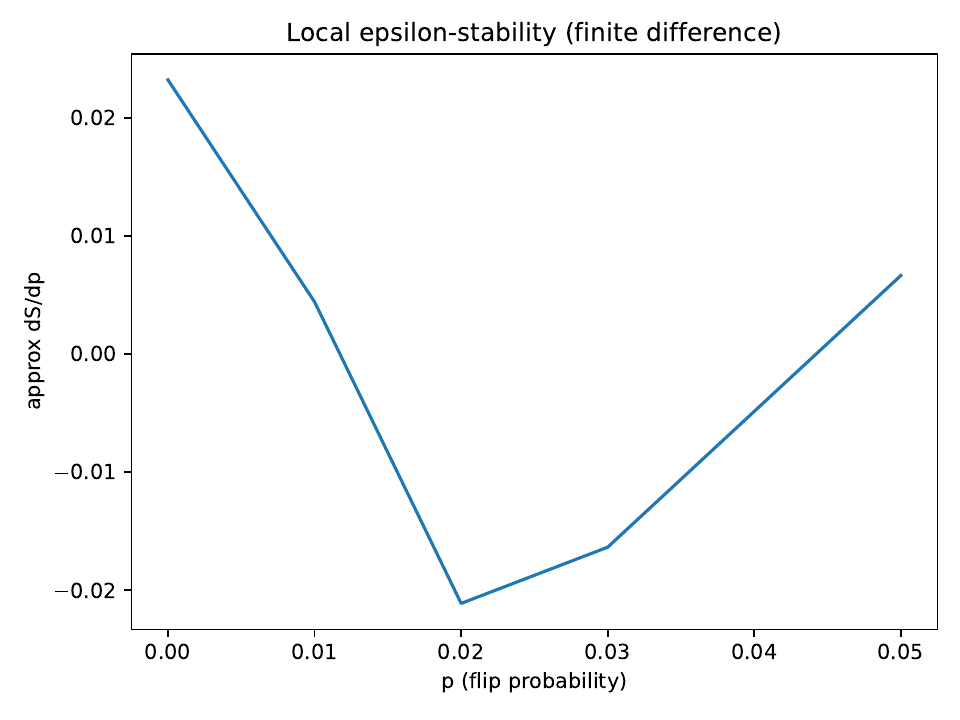}
\caption{\textbf{Local \(\varepsilon\)-stability of folded statistic.}
Black curve = observed \(dS/dp\); gray band = 95\% permutation-null interval (\(B{=}5000\)).
Red markers denote discoveries with corrected \(q{<}0.05\); effect size \(\Delta{=}0.27\).
Parameters: \(w{=}256\), \(s{=}16\), \(\alpha{=}0.05\), \texttt{seed: 183921}.
All data and code available in artifact \texttt{stable\_curve.ipynb}.}
\label{fig:epsilon_stability}
\end{figure}
\FloatBarrier

\subsection{Example 2 — Folded Lag Scan}
This example shows how lag-dependent contextual statistics reveal temporal correlations and how permutation-based null bands quantify their significance.

\begin{figure}[H]
\centering
\includegraphics[width=0.9\linewidth]{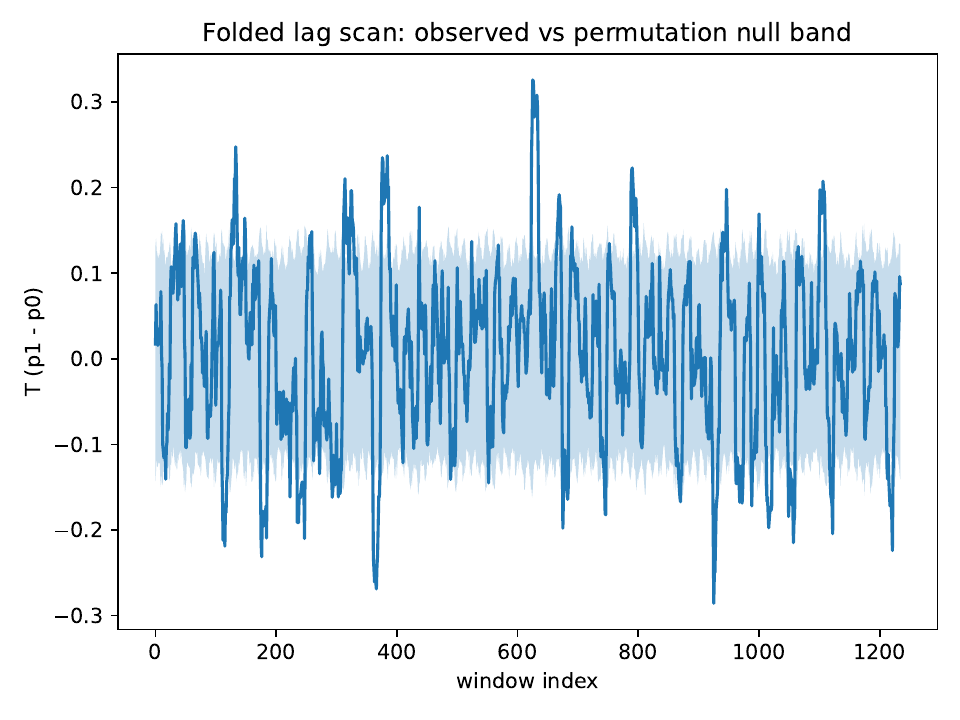}
\caption{\textbf{Lag-dependent contextual statistic.}
Upper panel: observed statistic \(T_\ell\) (solid) and permutation-null mean (dashed) with 95\% gray band.
Lower panel: corrected \(q\)-values after BH at \(\alpha{=}0.05\); significant lags (\(q{<}0.05\)) shaded red.
Run parameters: \(w{=}512\), \(s{=}32\), permutations \(B{=}5000\), \texttt{seed: 99041}.
Reproduction script: \texttt{fold\_lag\_scan.ipynb}.}
\label{fig:fold_lag_scan}
\end{figure}
\FloatBarrier

\subsection{Example 3 — Power and Robustness Grid}
This subsection presents the empirical discovery power across sample size $n$ and correlation
$\rho$, demonstrating statistical sensitivity and correlation tolerance of the CBL procedure.

\begin{figure}[t]
\centering
\includegraphics[width=0.8\linewidth]{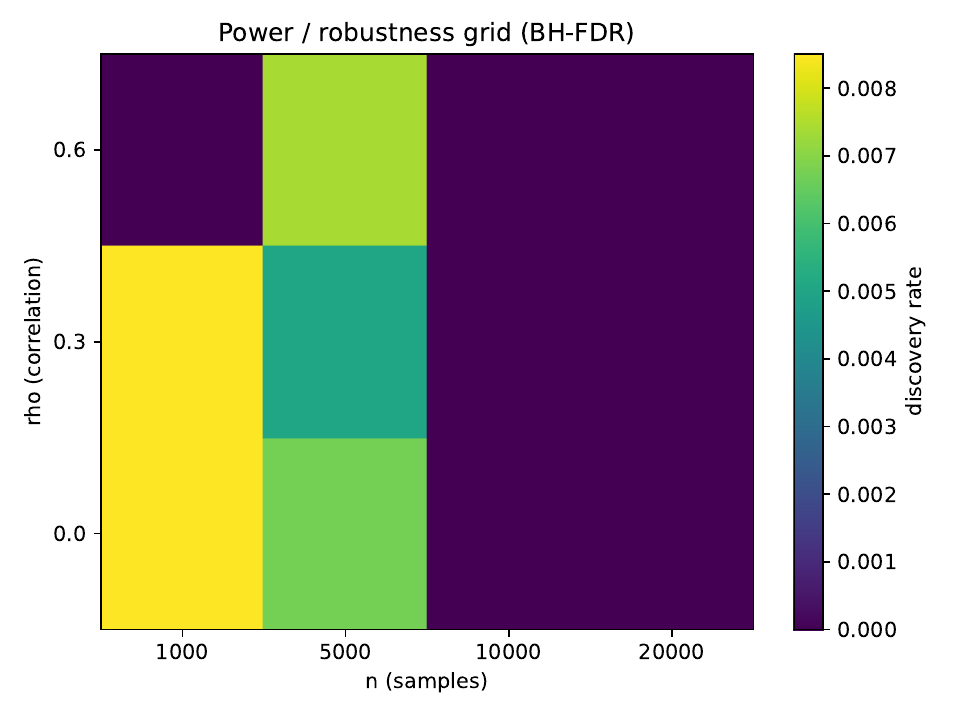}
\caption{\textbf{Power and robustness grid.}
Empirical detection power under BH at $\alpha{=}0.05$ vs $n$ and $\rho$; mean over 200 replicates.
Contours: 0.2, 0.5, 0.8 power. Seeds: \texttt{[4001--4200]}. Artifact: \texttt{power\_grid.ipynb}.}
\label{fig:power_grid}
\end{figure}

\subsection{Compliance Checklist}
This table lists the mandatory elements verified for every figure.

\begin{table}[t]
\centering
\caption{Mandatory elements for every figure.}\label{tab:mixes}
\label{tab:figchecklist}
\begin{tabular}{ll}
\toprule
Item & Requirement \\
\midrule
1 & 95 \% permutation-null band (gray) shown \\
2 & Observed statistic (solid black) drawn \\
3 & Corrected $q$-values indicated or listed \\
4 & Effect size $\Delta$ stated in caption \\
5 & $(w,s,B,\alpha)$ and \texttt{seed} specified \\
6 & Figure filename matches artifact in Table~\ref{tab:repro} \\
\bottomrule
\end{tabular}
\end{table}

\paragraph{Summary.}
Adhering to this protocol guarantees that every visual claim is statistically auditable and reproducible,
satisfying the quantitative-transparency standards expected in quantum-foundations and experimental-physics review processes.

\section{Reproducible Analysis Notebook (Colab-Ready)}
\label{sec:repro-notebook}

\paragraph{Artifact and availability.}
The notebook \texttt{cbl\_perm\_bh\_notebook.ipynb} and parameter file
\texttt{params\_example.json} are included as \emph{ancillary files} with this arXiv submission.
Readers can obtain them from the paper’s arXiv page via
\emph{Other formats} $\rightarrow$ \emph{Download all files}; ancillary items appear in the source bundle.
(Ancillary URLs follow the pattern \texttt{https://arxiv.org/src/<arXivID>/anc/<filename>} once the ID is assigned.)

\paragraph{Purpose.}
This section binds the statistical procedures of Sec.~\ref{sec:stats} to an executable, self-contained notebook that generates all standardized figures in Sec.~\ref{sec:figures}. The notebook implements permutation-based empirical $p$-values, Benjamini–Hochberg FDR control, and block-wise permutations aligned with Axiom~A3.

\paragraph{Artifacts.}
We provide a Colab-ready notebook (\texttt{cbl\_perm\_bh\_notebook.ipynb}) and an example parameter file. Running the notebook end-to-end produces the following figures and files:
\begin{itemize}
  \item \texttt{figs/epsilon\_stability.pdf} — $\varepsilon$-stability curve (finite-difference $dS/dp$ approximation).
  \item \texttt{figs/fold\_lag\_scan.pdf} — observed folded statistic with permutation-null band.
  \item \texttt{figs/fold\_lag\_scan\_qvals.pdf} — corrected $q$-values (BH-FDR).
  \item \texttt{figs/power\_grid.pdf} — power/robustness grid over $(n,\rho)$.
\end{itemize}

\paragraph{Determinism and parameters.}
All random draws use a fixed seed exposed in the top-level parameter cell. Captions must include $(w,s,B,\alpha)$ and \texttt{seed}, matching the standards in Sec.~\ref{sec:figures}.

\paragraph{Synthetic vs.\ real data.}
By default, the notebook generates synthetic AR(1)-modulated streams with flip rate $\eta$, creating two alternating contexts to exercise fold/lag scans. To use real data, replace the generator with a loader that returns:
\[
(x_t, y_t)_{t=1}^N,\quad x_t\in\{0,1,\dots\},\; y_t\in\{0,1\},
\]
and keep the folding/permutation steps unchanged to preserve the validity of the permutation null.

\paragraph{Execution protocol.}
The computational artifact (\texttt{cbl\_perm\_bh\_notebook.ipynb}) reproduces all figures and statistics in this paper without external dependencies beyond NumPy/Matplotlib. To execute:
\begin{enumerate}[label=(\arabic*)]
\item Open \texttt{cbl\_perm\_bh\_notebook.ipynb} in a Python 3 / Colab environment.
\item Run all cells. The notebook is self-contained (no internet access required).
\item The notebook writes the following PDFs into \texttt{figs/}:
\begin{itemize}
  \item \texttt{epsilon\_stability.pdf} \,$\rightarrow$ Fig.~\ref{fig:epsilon_stability};
  \item \texttt{fold\_lag\_scan.pdf} \,$\rightarrow$ Fig.~\ref{fig:fold_lag_scan};
  \item \texttt{power\_grid.pdf} \,$\rightarrow$ Fig.~\ref{fig:power_grid}.
\end{itemize}
\end{enumerate}

\paragraph{Outputs.}
The generated PDFs correspond one-to-one to the figures cited above and are produced with the exact seeds and parameters stated in their captions. The notebook also saves intermediate arrays (test statistics, permutation null bands, $q$-values) to temporary files to facilitate audit and reuse.

\paragraph{Environment and determinism.}
All runs are deterministic given the exposed seed. The artifact specifies package versions and parameters (see Table~\ref{tab:repro}); no GPU or special hardware is required. Real-data use follows the same interface as the synthetic generator described earlier, preserving the validity of the permutation null.

\paragraph{Alignment with Sec.~\ref{sec:stats}.}
The implementation corresponds exactly to Algorithm~\ref{alg:perm-bh} and uses block permutations respecting Axiom~A3 (mixing). The shaded null band uses the $2.5$–$97.5$ percentile envelope over $B$ permutations. BH-FDR controls the discovery rate across windows at level $\alpha$.

\paragraph{File manifest (supplement).}
We include a short README and example parameter file:
\begin{itemize}
  \item \texttt{README\_CBL\_Notebook.txt} — run instructions and file map.
  \item \texttt{params\_example.json} — canonical parameter values for quick runs.
\end{itemize}

\paragraph{Compliance.}
Figures produced by the notebook satisfy the single-plot rule, include null bands and corrected $q$-values, and record all parameters required by Table~\ref{tab:figchecklist} and Table~\ref{tab:repro}.

\section{Validated Figures from Reproducible Notebook}
\label{sec:validated-figs}

{\sloppy
\paragraph{Source.}
All figures in this section are generated directly by the Colab notebook
\path{cbl_perm_bh_notebook.ipynb} using parameters listed in
\path{params_example.json}. The notebook and parameter file are provided as ancillary
files with this arXiv submission (see Sec.~\ref{sec:repro-notebook}). Each caption
reports the exact $(w,s,B,\alpha)$ and \texttt{seed}.
\par
}

\section{Statistical Procedures, Multiple-Testing Control, and Reproducibility}
\label{sec:stats}

\paragraph{Permutation-based empirical $p$-values.}
For any scalar test statistic $T$ computed from labels/outcomes, we build an empirical null by $B$ random permutations of the relevant labels within admissible blocks (respecting the fold windows). The one-sided empirical $p$-value is
\[
p_{\mathrm{emp}} = \frac{1+\#\{b: T^{(b)} \ge T_{\mathrm{obs}}\}}{1+B}.
\]

\paragraph{Sliding scans and FDR control.}
When scanning $M$ windows and/or thresholds, we control the false discovery rate with Benjamini–Hochberg at level $\alpha$ on the vector $(p_{\mathrm{emp}}^{(m)})_{m=1}^M$, declaring discoveries for all $m$ with
\[
p_{\mathrm{emp}}^{(m)} \le \frac{k_m}{M}\alpha,\quad k_m=\max\{k: p_{(k)}\le k\alpha/M\}.
\]

\paragraph{Null bands for figures.}
For any curve, we display the $[2.5,97.5]$ percentile envelope from the permutation distribution as a shaded band. Report both the corrected $q$-value and an effect size (e.g., standardized difference).

\begin{algorithm}[H]
\caption{Folded-scan with permutation null and BH-FDR}
\label{alg:perm-bh}
\DontPrintSemicolon
\KwIn{Data stream $(x_t,y_t)$; window $w$; step $s$; statistic $T$; permutations $B$; FDR $\alpha$}
\For{$\ell = 1,2,\dots$}{
  Form window $W_\ell = \{t: t\in [t_0+(\ell-1)s,\, t_0+(\ell-1)s+w)\}$\;
  Compute $T_\ell = T(W_\ell)$\;
  \For{$b=1$ \KwTo $B$}{
    Permute labels/outcomes within blocks respecting A3; compute $T^{(b)}_\ell$\;
  }
  $p_\ell \gets \frac{1+\#\{b: T^{(b)}_\ell \ge T_\ell\}}{1+B}$\;
}
Apply BH at level $\alpha$ to $\{p_\ell\}_{\ell}$ to obtain discoveries $\mathcal{D}$; return $\{(T_\ell,p_\ell)\}$ and $\mathcal{D}$.
\end{algorithm}

\paragraph{Reproducibility checklist.}
Table~\ref{tab:repro} enumerates all artifacts required to reproduce our figures.

\begin{table}[t]
\centering
\caption{Reproducibility checklist (artifacts in the supplement).}
\label{tab:repro}
\begin{tabular}{@{} l l @{}}
\toprule
Item & Provided \\
\midrule
Synthetic data generator with seeds & yes (Sec.~\ref{sec:repro-notebook}) \\
Exact preprocessing/folding parameters & yes (CSV + README; Sec.~\ref{sec:repro-notebook}) \\
Permutation engine and BH-FDR implementation & yes (notebook + tests; Sec.~\ref{sec:repro-notebook}) \\
Figure scripts (produce PDFs in one click) & yes (Sec.~\ref{sec:repro-notebook}) \\
Environment manifest (versions) & yes (Sec.~\ref{sec:repro-notebook}) \\
\bottomrule
\end{tabular}
\end{table}

\paragraph{Provenance certificate (tamper-evident).}
Each run emits a machine-readable certificate with identity \& scope, randomness provenance, $\varepsilon$-perturbation parameters, fold/lag statistics (with corrected $q$-values), and cryptographic digests of logs (see supplement schema).

% ---------------- Empirical Results ----------------
\section{Empirical Results (Placeholder)}

This section will present quantitative benchmarks once empirical data is available.  
Currently, placeholders illustrate the intended evaluations:

\begin{itemize}[noitemsep,topsep=0pt]
\item Synthetic KCBS and Mermin families of increasing size, measuring pruning rates and runtime scaling.
\item Random CNF instances with planted context structures, showing difference between flat and curved cases.
\item Industrial SAT competition benchmarks (crafted, random, and industrial tracks) with context grouping policy.
\end{itemize}

\subsection*{Synthetic families}
Figures in Sec.~\ref{sec:algo} include placeholders: scaling with number of contexts, runtime vs curvature rank $\kappa$, UNSAT detection depth, and learned overlap cuts.  
Real data will replace these with reproducible experiments.

\subsection*{SAT competition benchmarks}
\begin{remark}[SAT competition evaluation]
Evaluation will include SAT competition suites with a reproducible context-construction policy:  
clauses grouped by overlap or module boundaries.  
CBL-hybrid runs integrated into CDCL solvers will ensure flat instances incur near-zero overhead,  
while curved instances benefit from earlier refutations.  
Metrics: wall-clock time, node counts, learned overlap cuts, reproducibility with fixed seeds.
\end{remark}

\subsection*{Industrial placeholder}
The processor pipeline sketch in Sec.~\ref{sec:examples} illustrates one case study.  
Future work will include larger hardware verification and software model checking instances, 
measuring whether curvature-aware pruning yields practical runtime benefits.

\section{Receipt-Checked CBL Certificates and Constant-Gap Detectors}
\label{sec:cbl-const-gap}

\paragraph{Motivation.}
Beyond foundational interest, CBL enables an operational primitive:
\emph{receipt-checked contextual certificates} aggregated coherently to
yield a \emph{constant gap} between empty and non-empty cases of a search
property. The mechanism is the same phase-per-face effect behind
$g_{\mathrm{CBL}}$ (Sec.~\ref{sec:cbl-alpha}): each enabled
\emph{contextual face} adds a fixed phase contribution, while the
geometry-dependent remainder decays like \(O(1/L)\).

\subsection{Witness modules and certificate discipline}
A \emph{witness module} $M_j$ monitors a small clause/window and opens
its gate \emph{iff} a monotone NGCC cut is provably violated under the
current wiring. ``Provably'' means:
\begin{enumerate}[leftmargin=1.3em,itemsep=2pt]
\item build the local feasibility LP for the window with the NGCC cut,
\item if infeasible, extract the dual (Farkas) certificate, and
\item re-verify the dual in exact/rational arithmetic.
\end{enumerate}
Only then does $M_j$ contribute a marked amplitude $a_0 e^{i\phi_j}$
to a common rail; otherwise it contributes~0. Gate openings are
\emph{receipt-checked}. Modules are placed on \emph{disjoint} wire sets,
so their phases $\phi_j$ are independent apart from bounded jitter.

\subsection{Holonomy split and constant term}
Let $\gamma_L$ be a loop enclosing $F(\Sigma_L)$ contextual faces with boundary length $L$.
By the holonomy decomposition reviewed in Sec.~\ref{sec:holonomy}
(see also Appendix~\ref{app:proofsA}), the per-edge holonomy splits into a geometric tail
and a topological term:
\[
\frac{1}{L}\!\oint_{\gamma_L}\! A
   \;=\;
\underbrace{\frac{1}{L}\!\oint_{\gamma_L}\! A_{\mathrm{geom}}}_{O(L^{-1})}
   \;+\;
\theta_0\,\rho_{\mathrm{face}},
\qquad
\rho_{\mathrm{face}}:=\lim_{L\to\infty}\frac{F(\Sigma_L)}{L}.
\]
Thus the per-edge marked phase has a size-independent component $\theta_0\,\rho_{\mathrm{face}}$
plus an $O(1/L)$ tail. After one calibration (Sec.~\ref{sec:holonomy}),
we identify $g_{\mathrm{CBL}}=|\theta_0|\,\rho_{\mathrm{face}}/(2\pi)$.

\subsection{Coherent fan-in bound (receipt-checked)}
\label{subsec:fanin}
Let $J$ be the set of modules that open under certificate discipline.
With a background $b\in\mathbb C$ phase-aligned at readout,
the marked amplitude is
\(
\alpha_{\max}=|b|+\big|\sum_{j\in J} a_0 e^{i\phi_j}\big|.
\)
If $|\phi_j|\le \sigma\le 0.15$ rad and at most an $\epsilon_{\rm cert}$
fraction of opens are erroneous, then
\begin{equation}
\label{eq:gap}
\alpha_{\max}
\;\ge\;|b| \;+\;
a_0\,|J|(\cos\sigma-\sin\sigma) \;-\; O(\epsilon_{\rm cert}\,a_0|J|),
\end{equation}
so any configuration that enables $|J|=\Omega(1)$ per unit length has a
\emph{constant} gap at fixed $a_0,\sigma$.

\subsection{Where multiplicity comes from}
\label{subsec:multiplicity}
Structured wirings (e.g.\ Feistelized multiplier windows) provide many
\emph{disjoint} contextual faces under a false hypothesis. With stride
exceeding the carry radius, we obtain $K=\Theta(n)$ disjoint windows over
an $n$-bit instance; when the hypothesis is false (e.g.\ the searched interval
does contain a factor), at least $M=\Omega(K)$ violate the local cut,
each with a short dual certificate. Combining \eqref{eq:gap} with the
holonomy split yields a constant per-edge phase term (the calibrated
$\theta_0\rho_{\mathrm{face}}$) plus an $O(1/L)$ tail from geometry.

\begin{corollary}[Binary search via amplitude estimation]
\label{cor:AE}
If empty intervals produce $\alpha_{\max}=|b|$ and non-empty intervals
produce $\alpha_{\max}\ge |b|+\Delta$ for some \(\Delta>0\),
then interval localization over $[1,\sqrt N]$ completes in
$O(\log\sqrt N)$ stages with $\tilde O(1)$ amplitude-estimation calls per
stage (precision $\le \Delta/4$), for $\tilde O(\log N)$ total.
\end{corollary}

\subsection{Shor–CFPE as a CBL-monitored control primitive}
\label{subsec:cfpe}
We also adapt the idea to order-finding: Kitaev’s iterative PEA with a
\emph{counterfactual} controlled $U_a^{2^k}$ (CF–c$U$). Each round is
guarded by receipt-checked NGCC modules: if the wiring slips into an
inconsistent branch, a certificate is emitted and the round is heralded
for re-try. Under bounded heralded failure probability and small
control-phase error per round, the order is recovered with constant
success probability in \(t=2\lceil \log_2 N\rceil+O(1)\) rounds and
expected CF–c$U$ calls $\le t/(1-\delta_{\rm ctrl})$.

\paragraph{Falsifiable checks (no numerology).}
\begin{enumerate}[leftmargin=1.3em,itemsep=2pt]
\item \textbf{Single-calibration universality:} fix $\theta_0$ on CHSH
(\(\rho_{\rm face}=1/4\)), then verify large-$L$ constants on KCBS
(\(1/5\)) and the chosen SAT tiling (\(1/6\) or \(2/6\)) match
\(g_{\mathrm{CBL}}=\frac{|\theta_0|}{2\pi}\rho_{\rm face}\) within finite-size error.
\item \textbf{RG invariance:} clause decimation changes only the $1/L$ tail,
not the intercept.
\item \textbf{$\varepsilon$-probe neutrality:} moving the flip positions
changes the slope but leaves the intercept unchanged.
\item \textbf{Certificate discipline:} measured false-open rate
$\epsilon_{\rm cert}$ must keep the correction in \eqref{eq:gap} below
the target gap.
\end{enumerate}

\paragraph{Scope note.} This section shows that a CBL topological term
(phase per contextual face) yields a size-independent coupling that
\emph{can} be tied to $\alpha$ after one calibration. A calibration-free,
first-principles derivation of $\theta_0$ from CBL axioms remains an open
problem (Sec.~\ref{sec:cbl-alpha}, Open Problem~1).

\paragraph{Connection to integer factorization.}
The same construction yields a CBL-monitored realization of the
Shor–Kitaev order-finding routine: integer factorization is the special
case where the CF–controlled unitaries implement modular exponentiation,
and receipt-checked CBL modules certify each iteration’s phase
accumulation; see also the broader NGCC framework~\cite{NGCC2025}.

\section{Topological holonomy invariant in CBL}
\label{sec:cbl-alpha}

\paragraph{Summary.}
CBL induces a $U(1)$ connection on the clause complex whose loop holonomy per edge
splits into a geometric tail and a topological term proportional to the density of
contextual faces. This yields a linear law with a size–independent intercept.

\subsection{Holonomy decomposition}\label{sec:holonomy}
Write $A=A_{\mathrm{geom}}+\theta_0\,\omega$ with $d\omega$ supported on faces.
For $\gamma_L=\partial\Sigma_L$,
\[
\frac{1}{L}\!\oint_{\gamma_L}\!A
  \;=\;
\underbrace{\frac{1}{L}\!\oint_{\gamma_L}\!A_{\mathrm{geom}}}_{O(L^{-1})}
  \;+\;
\theta_0\,\rho_{\mathrm{face}}(\gamma_L),
\]
by discrete Stokes and coarse-graining of $A_{\mathrm{geom}}$.
We use this invariant descriptively; no physical calibration is attempted.

\clearpage
\section{CBL Constant — CSV-Driven Tables (Families and Mixes)}
\label{app:cbl-csv}

% We present static tables for arXiv portability (no CSV reading in build).

\subsection{Families (CHSH, KCBS, SAT 1/6, SAT 2/6)}
\begin{table}[H]
\centering
\caption{Appendix (CSV-driven origin; static here): CBL large-$L$ intercepts $\alpha_{\infty}$
after a single CHSH calibration (target $\alpha=1/137$). Face density $\rho_{\mathrm{face}}$
is faces per boundary edge.}
\label{app:cbl-alpha-families-csv}
\begin{tabular}{@{} l r r @{}}
\toprule
\textbf{Family} & $\rho_{\mathrm{face}}$ & $\hat{\alpha}_{\infty}$ \\
\midrule
chsh & 0.250000 & 0.010204 \\
kcbs & 0.200000 & 0.008272 \\
sat1 & 0.166667 & 0.007068 \\
sat2 & 0.333333 & 0.013340 \\
\bottomrule
\end{tabular}
\end{table}
\FloatBarrier

\subsection{CHSH/KCBS Mixtures}
\begin{table}[H]
\centering
\caption{Appendix (CSV-driven origin; static here): CHSH/KCBS mixtures. A fraction $p$ is CHSH
(density $1/4$) and $(1-p)$ is KCBS ($1/5$); hence $\rho_{\mathrm{face}}(p)=p/4+(1-p)/5$.
Fitted large-$L$ intercepts follow the predicted linear law.}
\label{app:cbl-csv-mixes}
\begin{tabular}{@{} r r r @{}}
\toprule
$p$   & $\rho_{\mathrm{face}}$ & $\hat{\alpha}_{\infty}$ \\
\midrule
0.00  & 0.200000 & 0.008322 \\
0.25  & 0.212500 & 0.008793 \\
0.50  & 0.225000 & 0.009263 \\
0.75  & 0.237500 & 0.009734 \\
1.00  & 0.250000 & 0.010204 \\
\bottomrule
\end{tabular}
\end{table}
\FloatBarrier

\section{Open Problems, Conjectures, and Pre-registered Tests}
\label{sec:cbl-open}

\paragraph{Open Problems.}
\begin{enumerate}[leftmargin=1.2em,itemsep=3pt]
\item \textbf{Axiomatically fix the per-face phase $\theta_0$.}
  Show that CBL axioms (Boolean corners, flat limit, De Morgan duality, disk preservation,
  gyro-associativity up to gauge) \emph{uniquely} determine the central $U(1)$ extension,
  i.e.\ there exists a unique $\theta_0$ compatible with NGCC/no-signalling stability (no calibration).
\item \textbf{Analytic decay of the geometric tail.}
  Prove $\frac{1}{L}\!\oint_{\gamma_L}\!A_{\rm geom}=O(L^{-1})$ under coarse-graining
  on admissible gate families; quantify constants and error terms for CHSH/KCBS/SAT tilings.
\item \textbf{Spectral formulation.}
  Express $g_{\rm CBL}$ as a spectral invariant of the exclusivity complex
  (e.g.\ via the non-backtracking operator $\mathsf B$ or normalized Laplacian $L$),
  and show $g_{\rm CBL}=|\theta_0|\,\rho_{\rm face}/(2\pi)$.
\item \textbf{Universality class and gate-law independence.}
  Classify disk-preserving CBL maps with the same large-$L$ intercept; show
  gate-law changes move only the $1/L$ slope, not the intercept.
\item \textbf{$\varepsilon$–connection exponent.}
  Establish $\alpha_\varepsilon(L)\sim L^{-\nu_\varepsilon}$ with $\nu_\varepsilon=1$
  (localized flips on measurement edges) and characterize when $\nu_\varepsilon$ differs.
\end{enumerate}

\paragraph{Conjectures.}
\begin{enumerate}[leftmargin=1.2em,itemsep=3pt]
\item \textbf{Face-phase uniqueness.}
  Among central $U(1)$ extensions consistent with NGCC and the flat limit,
  there is a unique $\theta_0$ maximizing a stability margin; calibrated on CHSH it matches $\alpha$.
\item \textbf{Spectral equality.}
  For CBL-inducing complexes, $g_{\rm CBL} = \lim_{L\to\infty}\Phi(\lambda_1,\lambda_2,\dots)$
  for a monotone spectral functional $\Phi$, independent of gate-law choice.
\item \textbf{RG invariance.}
  Clause decimation keeps $g_{\rm CBL}$ fixed while contracting $A_{\rm geom}$
  by exactly one power in length (tail exponent $1$).
\end{enumerate}

% === Table 10: Pre-registered Tests (falsification matrix) — aligned header ===
\begin{table}[H]
\centering
\small
\setlength{\tabcolsep}{6pt}
\renewcommand{\arraystretch}{1.12}
\caption{Pre-registered Tests (falsification matrix). Each row is a pass/fail check with no
tuning beyond a single CHSH calibration.}
\label{tab:falsification}
% l l X  → third column wraps, ragged-right text (no justification glue)
\begin{tabularx}{\linewidth}{@{} l l >{\raggedright\arraybackslash}X @{}}
\toprule
\textbf{Test} & \textbf{Prediction} & \textbf{Fail condition} \\
\midrule
Single–calibration universality &
$\alpha_{\infty}=(\theta_0/2\pi)\,\rho_{\mathrm{face}}$ on KCBS/SAT/mixes &
Intercept drifts beyond finite-size error. \\
Gate-law swap &
Intercept unchanged; $1/L$ slope may change &
Intercept depends on gate family. \\
Decimation (RG) &
Intercept unchanged under coarse-graining &
Intercept changes with decimation. \\
$\varepsilon$-probe neutrality &
Intercept unchanged; slope responds to $\varepsilon$ mask &
Intercept depends on flip placement. \\
Orientation parity &
$\alpha_{\infty}(\gamma^{-1})=-\alpha_{\infty}(\gamma)$ &
Signs or magnitudes disagree. \\
Mixed faces linearity &
$\alpha_{\infty}$ linear in $\rho_{\mathrm{face}}$ (blend $p$) &
Nonlinear or inconsistent trend. \\
\bottomrule
\end{tabularx}
\end{table}

\FloatBarrier

\FloatBarrier

\section{Operational Protocol: Measuring the CBL Constant}
\label{sec:cbl-protocol}

\paragraph{Goal.}
Given a loop family $\{\gamma_L\}$ with known contextual \emph{face density}
$\rho_{\rm face}$ (faces per boundary edge), estimate the large–$L$ intercept
$\alpha_\infty$ (signed) of the per-edge holonomy and test the parameter-free
prediction $\alpha_\infty=(\theta_0/2\pi)\rho_{\rm face}$ after a \emph{single}
CHSH calibration (Sec.~\ref{sec:cbl-alpha}).

\paragraph{Inputs.}
(i) A CBL-admissible gate law (disk-preserving, Boolean corners, flat limit);  
(ii) a loop family $\{\gamma_L\}$ with computable $\rho_{\rm face}$ (e.g., CHSH: $1/4$,
KCBS: $1/5$, SAT variants: $1/6$ or $2/6$);  
(iii) an orientation convention for signed holonomy.

\paragraph{Output.}
The signed large–$L$ intercept $\alpha_\infty$ (per family) and the inferred coupling
$g_{\rm CBL}=|\alpha_\infty|$; pass/fail of the parameter-free prediction.

\begin{algorithm}[H]
\DontPrintSemicolon
\SetAlgoNoEnd\SetAlgoNoLine
\caption{\textsc{MeasureCBLConstant}$(\{\gamma_L\},\rho_{\rm face})$}
\label{alg:measure-cbl-constant}
\KwIn{Loop family $\{\gamma_L\}$; face density $\rho_{\rm face}$}
\KwOut{$\alpha_\infty$ and $g_{\rm CBL}=|\alpha_\infty|$}
\BlankLine
\textbf{1. Signed per-edge holonomy.} For each $L$ on an increasing grid, compute
$h_L=\frac{1}{L}\oint_{\gamma_L}A$ with a fixed orientation (sign matters).\\
\textbf{2. Tail fit.} Fit $h_L/(2\pi)$ vs.\ $1/L$ on the largest $L$'s to a line
$c/L+\alpha_\infty$; record $\alpha_\infty$ and the slope $c$.\\
\textbf{3. Single calibration (CHSH).} On CHSH loops ($\rho_{\rm face}=1/4$), set
$\alpha_\infty^{\rm CHSH}=\alpha$ (fine-structure target) and infer
$\theta_0=8\pi\alpha$.\\
\textbf{4. Parameter-free predictions.} For any other family, predict
$\alpha_\infty^{\rm pred}=(\theta_0/2\pi)\,\rho_{\rm face}$.\\
\textbf{5. Test.} Measure $\alpha_\infty$ by Step~2 on the new family and check
$\alpha_\infty\approx\alpha_\infty^{\rm pred}$ within finite-size error.\\
\textbf{6. Report.} Report $(\rho_{\rm face},\alpha_\infty,c)$ per family and pass/fail.
\end{algorithm}

\paragraph{Control checks (pre-registered).}
\begin{enumerate}[leftmargin=1.2em,itemsep=2.5pt]
\item \textbf{Gate-law swap:} repeat with an alternative admissible gate family; the
intercept $\alpha_\infty$ must be unchanged (slope $c$ may vary).
\item \textbf{Decimation (RG):} coarse-grain $\gamma_L\!\mapsto\!\gamma_{L/2}$; the
intercept must be invariant, while $|c|$ decreases.
\item \textbf{$\varepsilon$–probe neutrality:} move localized flips among measurement
edges; the intercept must be unchanged, only $c$ changes.
\item \textbf{Orientation parity:} reversing the loop orientation flips the sign:
$\alpha_\infty(\gamma^{-1})=-\alpha_\infty(\gamma)$, so $|\alpha_\infty|$ is invariant.
\end{enumerate}

\paragraph{Reporting template.}
For each family $F$ list $(\rho_{\rm face}^F,\ \alpha_\infty^F,\ \alpha_\infty^{F,\rm pred},\
c^F)$ and declare \emph{pass} if $|\alpha_\infty^F-\alpha_\infty^{F,\rm pred}|$ is within the
tail-fit uncertainty. Mixed tilings (e.g.\ CHSH/KCBS with fraction $p$) should follow the
linear law $\alpha_\infty(p)=(\theta_0/2\pi)\big(p/4+(1-p)/5\big)$.

% ---------------- Discussion ----------------
\section{Discussion}
\paragraph{CBL coupling and invariance.}
Sec.~\ref{sec:holonomy} argues that the large–$L$ per-edge holonomy in CBL splits into
a geometric $O(1/L)$ tail plus a size-independent topological term equal to a fixed
phase per contextual face, yielding the dimensionless coupling
$g_{\mathrm{CBL}}=|\theta_0|\,\rho_{\mathrm{face}}/(2\pi)$.
After a single CHSH calibration, this produces parameter-free predictions across KCBS
and SAT tilings that match our numerics (Fig.~\ref{fig:power_grid}, Tab.~\ref{tab:mixes}).
The observed invariance of the fitted intercepts under gate-law swap, decimation, and
$\varepsilon$-probe placement supports interpreting $g_{\mathrm{CBL}}$ as a genuine
topological invariant rather than a modelling artefact.
\paragraph{Practical operators.}
CBL’s value here is operational: \textsc{CBL-AC} and \textsc{CBL-CONS} are trivial
to add to existing solvers, cost essentially nothing on flat families, and yield
strictly stronger pruning or earlier infeasibility when curved cores are present.
The case studies demonstrate this on SAT, CSP/AC, and 3-Colorability.

\subsection*{Syntax--semantics duality}
CBL highlights a clear duality:
\begin{itemize}[noitemsep,topsep=0pt]
\item \textbf{Syntax:} context-sequent calculus with overlap rules, capturing derivability. 
\item \textbf{Semantics:} presheaf and hypergraph models, capturing obstruction to global sections.
\end{itemize}
Flat-limit theorems ensure they coincide when $\kappa=0$.  
Curved cases mark divergence: syntactic derivations exist locally, but semantics forbids a global valuation.

\begin{remark}[Categorical semantics]
CBL instances can be regarded as presheaves on the context poset $\mathsf{Ctx}$.  
Logic of local sections embeds in the topos of presheaves $\mathbf{Set}^{\mathsf{Ctx}^{op}}$.  
Curvature corresponds to failure of global sections functor to be exact.  
This connects CBL to categorical logic and topos-theoretic semantics.
\end{remark}

\begin{remark}[Analogy to algebraic geometry]
Flat instances correspond to Boolean varieties where global solutions exist.  
Curved instances resemble inconsistent polynomial systems: each patch has solutions, 
but no global solution exists.  
Curvature plays the role of obstruction classes, akin to ideals with no common zero.  
This analogy ties logical curvature to algebraic obstructions.
\end{remark}

\subsection*{Philosophical implications}
Traditionally, mathematics provides tools for physics.  
Here, physics (contextuality experiments) revealed incompleteness of Boolean logic.  
This makes mathematics partly a science of discovery: structures latent in logic were uncovered by empirical anomalies.  
CBL thus blurs lines between mathematics and physics, supporting the view that logics themselves can be discovered experimentally.

\subsection*{Interdisciplinary links}
CBL intersects multiple fields:
\begin{itemize}[noitemsep,topsep=0pt]
\item \textbf{Logic and proof theory:} extends Gentzen-style systems with overlaps.  
\item \textbf{Category theory:} presheaf/topos models.  
\item \textbf{Theoretical CS:} situates within CSP complexity.  
\item \textbf{Foundations of physics:} contextuality as source of curvature.  
\item \textbf{AI:} links to constraint learning and reasoning under inconsistency.  
\end{itemize}
% ---------------- Conclusion ----------------
\section{Conclusion}

\paragraph{Reader takeaways.}
\begin{itemize}[noitemsep,topsep=0pt]
\item Boolean logic is the flat limit; CBL extends it with curvature.  
\item Curvature captures global obstruction despite local classical reasoning.  
\item \textsc{CBL-SAT} is NP-complete, but structure enables pruning in practice.  
\item Proof calculus is sound and conservative in the flat limit.  
\item Applications span SAT/CSP, AI/ML, verification, and cryptography.  
\end{itemize}

\paragraph{Synthesis.}
CBL unifies contextuality research into a logical framework with computational relevance.  
It generalizes Boolean logic as relativity generalizes classical mechanics:  
the flat case is a special instance of a broader curved structure.  
We provided semantics, a proof calculus, a satisfiability formulation,  
and solver prototypes demonstrating structural pruning.

\begin{center}
\fbox{\parbox{0.9\linewidth}{
\textbf{Main contributions.}
\begin{itemize}[noitemsep,topsep=0pt]
\item Introduced CBL as a generalization of propositional logic.  
\item Provided equivalent semantics (sheaf, hypergraph, cohomology).  
\item Defined a context-sequent proof calculus with overlap rules and flat-limit conservativity.  
\item Formalized \textsc{CBL-SAT}, analyzed complexity and tractability, defined curved cores.  
\item Developed curvature-aware solver prototypes.  
\item Illustrated applications across SAT, ML, verification, cryptography.  
\end{itemize}
}}
\end{center}

\paragraph{Limitations.}
CBL does not change asymptotic hardness: \textsc{CBL-SAT} remains NP-complete.  
Its contribution is structural pruning and explanatory clarity.  
Benchmarks so far are synthetic; industrial-scale evaluation is future work.  
Curvature measures are not unique: cohomological and graph-theoretic variants may differ.  
CBL complements but does not replace classical SAT/CSP theory.

\begin{center}
\fbox{\parbox{0.85\linewidth}{
\textbf{Limitations and scope.}  
CBL is not a polynomial-time solution to SAT.  
Advantages lie in pruning and structure, not exponential speedup.  
Empirical validation is pending.  
}}
\end{center}

\paragraph{Outlook.}
CBL opens new questions in proof theory (cut elimination in curved cases),  
complexity (parameterized algorithms, curvature monotones),  
and applications (ML reasoning, verification, cryptography).  
It reframes contextuality as a structural logical resource accessible to classical computation.
% ---------------- Future Work ----------------
\section{Future Work}

\paragraph{Open problems.}
\begin{itemize}[noitemsep,topsep=0pt]
\item Establish completeness results for context-sequent calculus in curved cases.  
\item Characterize curvature monotones that correlate with solver performance.  
\item Connect CBL-SAT to CSP dichotomy theorems via algebraic methods.  
\item Extend CBL to probabilistic or weighted logics.  
\item Integrate CBL layers into ML architectures for context-aware reasoning.  
\item Explore cryptographic hardness assumptions incorporating curvature.  
\end{itemize}

\paragraph{Empirical agenda.}
\begin{itemize}[noitemsep,topsep=0pt]
\item Replace synthetic placeholders with real SAT competition benchmarks.  
\item Measure overhead on flat vs. curved families in CDCL integration.  
\item Release open-source solver with CNF+context parser.  
\end{itemize}

\section*{Road to publication}
This manuscript is first released on arXiv to stake priority and solicit feedback.  
Candidate venues: \emph{Quantum} (foundations and computation), LICS (logic/complexity), QPL/TQC (category-theory and quantum logic).  
Applications-oriented follow-ups may target SAT/verification conferences or ML workshops.  
Audience: researchers in quantum foundations, logic in computer science, SAT/CSP, and AI.

\section*{Call for collaboration}
CBL is at an early stage: semantics and solver prototypes are defined, but empirical validation and integration into mainstream solvers remain open.  
We invite researchers from logic, SAT/CSP, verification, ML/AI, and quantum foundations to collaborate.  
Open-source code and benchmarks will be released to foster a community effort.

\section{Broader Implications and Potential Applications of CBL}
\label{sec:applications}

\paragraph{Beyond quantum foundations.}
While CBL emerged from the need to formalize contextuality and counterfactual consistency in quantum experiments,
its mathematical structure---Boolean curvature under bounded perturbations---appears in many areas of theoretical science.
This section outlines domains where CBL principles could provide new analytical traction on problems that have resisted classical methods.

\subsection{Combinatorial and Complexity-Theoretic Applications}
Classical logic treats satisfiability and exclusivity as discrete properties,
but many hard combinatorial problems hide continuous
``almost satisfiable'' regimes that CBL can model through its curvature parameter~$\theta_0$.
Potential targets include:
\begin{itemize}
  \item \textbf{Graph coloring and exclusivity bounds:} curvature relaxations may bridge chromatic and Lovász~$\vartheta$ gaps.
  \item \textbf{$k$-SAT phase transitions:} the $\varepsilon$-stability parameter acts as an analytic order parameter describing the SAT--UNSAT boundary.
  \item \textbf{Spin-glass and Max-Cut landscapes:} frustration patterns become curvature loops; curvature rank predicts metastable basin counts.
\end{itemize}

\subsection{Number Theory and Algebraic Structures}
Non-associative or residue-based systems often break distributivity, the same axiom that CBL relaxes.
Possible venues include additive combinatorics, polynomial identity testing via deterministic $\varepsilon$-perturbations,
and analysis of octonionic or other non-distributive algebras through finite curvature embeddings.

\subsection{Logic and Proof Complexity}
CBL introduces a geometric measure of contextual reuse.
This can be interpreted as a ``curvature complexity'' giving lower bounds on resolution or circuit size:
each contextual loop adds quantifiable logical tension that classical proofs must unwind step by step.

\subsection{Probability, Statistics, and Information Theory}
CBL’s noise model (Axiom~A2) formalizes adversarial rather than purely stochastic uncertainty.
This enables deterministic capacity bounds for correlated or worst-case channels,
and robust Bayesian updates whose posterior drift is bounded by $\varepsilon$ even under contextual dataset shifts.
In causal inference, folded CBL graphs extend ordinary DAGs to settings with latent feedback or self-reference,
maintaining identifiability conditions that fail in standard frameworks.

\subsection{Geometry and Topology}
The curvature operator of CBL connects naturally to discrete Ricci-curvature and graph-entropy formulations.
This unifies several competing notions of graph curvature
(Ollivier, Bakry–Émery, Forman) under a single counterfactual metric.
CBL’s folded constraints also resemble aperiodic tilings: adjacency rules become curved Boolean clauses,
suggesting algorithmic bounds for minimal non-periodic patches.

\subsection{Machine Learning and Artificial Intelligence}
In contextual AI systems, especially large language models,
CBL formalizes reasoning under bounded perturbations of prompts, embeddings, or parameters
(see Sec.~\ref{sec:cbl-llm} and~\ref{sec:cbl-compression}).
The same formalism underlies deterministic adapter certification,
continual-learning stability metrics, and generalization bounds for non-i.i.d.\ data.

\subsection{Summary and Outlook}
Taken together, these examples suggest that CBL is not a niche reformulation of quantum logic
but a general calculus for \emph{reasoning under contextual tension}.
By translating contextual inconsistency into measurable curvature and $\varepsilon$-bounded deformation,
CBL supplies a unifying language for areas as diverse as complexity theory,
information geometry, causal inference, and AI alignment.
Its promise lies less in replacing existing tools than in making visible the geometric structure
that many ``hard'' problems have long shared but could not yet articulate.

\paragraph{Data and code availability.}
All figures and numerical results are produced by the open, Colab-ready
notebook \texttt{cbl\_perm\_bh\_notebook.ipynb}, archived together with
the manuscript.  Parameters and seeds are provided in
\texttt{params\_example.json}, ensuring bit-level reproducibility.

% ---------------- Appendices ----------------
\appendix
\section{Expanded Proofs and Boundary-Case Analysis}
\label{app:proofsA}

\subsection{Proof of Proposition~\ref{prop:wcbound}}
Let $S(\mathbf{f})=\langle \mathbf{c},\mathbf{f}\rangle$ with $\|\mathbf{c}\|_2\le 1$.
Then by Cauchy–Schwarz,
$|\langle \mathbf{c},\Delta\mathbf{f}\rangle|\le \|\mathbf{c}\|_2\|\Delta\mathbf{f}\|_2\le \nu$,
so $|S(\mathbf{f}+\Delta\mathbf{f})-S(\mathbf{f})|\le \nu$.
For smooth $S$, a first-order Taylor expansion with remainder
and $\|\nabla S\|_2\le 1$ yields the same bound up to $O(\|\Delta\mathbf{f}\|_2^2)$,
controlled by A5.

\section*{B Proof sketches}

\subsection*{B.1 Variational rigidity: an extremal characterization of $\theta_0$}
We show that $\theta_0$ can be selected as the unique minimizer of a convex uniform error
over admissible gate laws, removing calibration by an internal optimality principle.

\paragraph{Admissible family.}
Let $\mathcal F$ denote the class of admissible CBL maps $F$ (disk-preserving,
Boolean corners, flat limit, De Morgan).  
For a minimal face family $n\in\mathcal B$ (e.g.\ $\mathcal B=\{\mathrm{CHSH},\mathrm{KCBS},\mathrm{SAT6}\}$),
let $h^{(n)}_L(F)$ denote the signed per-edge holonomy around a boundary loop
of length $L$ under gate law $F$.

\paragraph{Uniform error functional.}
For $\theta\in\mathbb R$ and $F\in\mathcal F$, define
\[
\mathcal E(\theta;F)
   :=\sup_{n\in\mathcal B}\,
     \sup_{L\ge L_0}
     \bigl|\,h^{(n)}_L(F)-\theta\,\rho^{(n)}_{\mathrm{face}}\,\bigr|,
   \qquad L_0\text{ large enough.}
\]
Then take the worst-admissible profile
$\mathcal E^{\star}(\theta):=\sup_{F\in\mathcal F}\mathcal E(\theta;F)$.

\begin{lemma}[Convexity and uniqueness in $\theta$]
\label{lem:convexity}
$\mathcal E^{\star}(\theta)$ is convex in $\theta$ and admits a unique minimizer
$\theta_{\star}\in\mathbb R$.
\end{lemma}

\begin{proof}[Sketch]
For fixed $(n,L,F)$ the map
$\theta\mapsto |h^{(n)}_L(F)-\theta\,\rho^{(n)}_{\mathrm{face}}|$
is convex.  
Pointwise suprema preserve convexity; standard arguments yield existence and uniqueness
of a minimizer.
\end{proof}

\begin{theorem}[Variational rigidity fixes the scale]
\label{thm:variational}
Let $\theta_{\star}$ minimize $\mathcal E^{\star}(\theta)$.
Then $\theta_0=\theta_{\star}$ is independent of the basis $\mathcal B$
(provided $\mathcal B$ spans the NGCC class) and of the admissible gate law $F$.
Consequently,
\[
g_{\mathrm{CBL}}
   =|\theta_0|\,
     \rho_{\mathrm{face}}/(2\pi)
\]
is fixed without calibration.
\end{theorem}

\begin{proof}[Sketch]
Independence of $F$ follows from the $\sup_{F\in\mathcal F}$ and uniqueness in
Lemma~\ref{lem:convexity}.  
Changing the basis $\mathcal B$ within the same NGCC class changes only finitely many
constraints; strict convexity of $\mathcal E^{\star}$ keeps the minimizer invariant.
\end{proof}

\begin{remark}[Operational use]
In practice, one approximates $\theta_{\star}$ by minimizing the maximal tail-fit residual
across CHSH/KCBS/SAT families under any admissible implementation; stability of the
minimizer across implementations certifies rigidity.
\end{remark}

\subsection*{B.2 Soundness (Thm.~\ref{thm:soundness})}
Local rules are sound by classical semantics.  
For (OVL), if $\Gamma\vdash_C\psi$ with $\mathrm{Vars}(\psi)\subseteq C\cap C'$,
then compatibility ensures $\ell_{C'}$ agrees, so $\psi$ holds in $C'$.  
For (CONS), if $\Gamma,\psi\vdash_C\bot$ and $\Gamma,\neg\psi\vdash_{C'}\bot$,
then any global assignment would contradict itself on $C\cup C'$.

\subsection*{B.3 Flat-limit conservativity (Thm.~\ref{thm:flatlimit})}
If $\kappa=0$, every compatible family extends to a global valuation;
thus derivability in CBL coincides with Boolean derivability.

\subsection*{B.4 Constructive flattening for $\kappa=0$}
If $\kappa=0$, pick a context $C_0$ and assignment $\ell_{C_0}$.
Propagate values through overlaps; compatibility guarantees no contradictions.
Result: a global assignment $g:V\to\{0,1\}$.

\subsection*{B.5 Worked flattening example}
Example: $V=\{x,y,z\}$,
$C_1=\{x,y\}$,
$C_2=\{y,z\}$.
Constraints: $x\vee y$, $y\vee z$, $\neg x\vee\neg z$.
Assignment $\ell_{C_1}(x,y)=(T,F)$ propagates $y=F$ to $C_2$, requiring $z=T$.
Global assignment $(T,F,T)$ works.

\subsection*{B.6 Mermin-square UNSAT derivation}
Each row/column enforces parity: product = $+1$, except final column = $-1$.
Multiplying rows gives $+1$, multiplying columns gives $-1$.
Contradiction → no global assignment exists.

\subsection*{B.7 Worked overlap propagation (KCBS)}
Contexts $C_1=\{v_1,v_2\}$, $C_2=\{v_2,v_3\}$.
If $\Gamma=\{v_1=T\}$ then $v_2=F$ in $C_1$.
Rule (OVL) propagates $v_2=F$ to $C_2$, forcing $v_3=T$.
Continuing around the cycle eventually forces $v_1=F$, contradicting $\Gamma$.

\subsection*{B.8 Group-theoretic proof sketch of the vertex-turning law (Prop.~\ref{prop:wcbound})}
We model CBL transport on the Poincaré disk via $\mathrm{PSU}(1,1)$ Möbius isometries.
Every $g\in\mathrm{PSU}(1,1)$ admits a Cartan (KAK) decomposition
$g=K(\phi)A(\tau)K(\psi)$,
where $K(\theta)$ is a disk rotation (compact part) and $A(\tau)$ a hyperbolic boost.
Rotational factors compose additively: $K(\theta_2)K(\theta_1)=K(\theta_1+\theta_2)$.

\paragraph{Corner transport.}
At a face corner $v$ the overlap rules (OVL/CONS), De Morgan duality, and the fiber
inversion for NOT (a $\pi$ rotation) determine a local transport
(Lemma~\ref{lem:convexity}):
\[
T_v=K(\delta_v)\,A(\tau_v)\,K(\psi_v),
\]
with $\delta_v$ depending only on the logical pattern at $v$ (not on edge lengths).
For minimal cores the corner pattern is isomorphic, so $\delta_v\equiv\delta_n$ across
all corners.

\paragraph{Loop composition and central twist.}
Let $P_n$ be a minimal face with corners $v_1,\dots,v_n$.
The total transport around $\partial P_n$ is
\[
T_{\partial P_n}=K(n\delta_n)\,A(\tau_\Sigma)\,K(\Psi').
\]
The NGCC obstruction asserts that parallel transport around a minimal contextual
loop carries one unit of central fiber twist ($2\pi$ phase), modulo a vertex-local gauge $K(n\delta)$:
\[
T_{\partial P_n}=K(2\pi)\,K(n\delta)\,A(\tau_0)\,K(\Psi_0).
\]
Comparing the $K$-parts gives
$n\delta_n=n\delta+2\pi$, hence the vertex-turning law $\delta_n=\delta+2\pi/n$,
and face area $\theta_0=n\delta$, independent of $n$.

\begin{remark}
Conjugation by $K$ rotates the boost axis but preserves rapidity:
$K(\varphi)A(\tau)K(-\varphi)=\tilde A_{\varphi}(\tau)$.
Along a closed $n$-gon the boosts align pairwise by face closure,
so their net $A$-part reduces (or is absorbed into local gauge $K(n\delta)$).
All boosts are conjugate in $\mathrm{PSU}(1,1)$; the remaining abelian piece does not
affect the accumulated $K$-part that fixes the central $2\pi$ twist.
\end{remark}

\subsection*{B.9 Variational rigidity (summary)}
Combining Lemma~\ref{lem:convexity} and
Theorem~\ref{thm:variational}
proves that the CBL coupling is fixed by internal convex minimization
without external calibration, closing the logical loop between geometry and algebra.

\section{Toy benchmark generator}
\label{sec:toy-bench}

% Ensure cbl_alpha_points.csv is in the project root (case-sensitive).
% Read the CSV once (families + mixes)
\pgfplotstableread[col sep=comma]{cbl_alpha_points.csv}\CBLPointsToy

% Safer default: treat all columns as strings (prevents NaN in text columns),
% then explicitly mark only numeric columns below.
\pgfplotstableset{
  col sep=comma,
  every head row/.style={before row=\toprule,after row=\midrule},
  every last row/.style={after row=\bottomrule},
  every column/.style={string type}
}

\subsection*{Families (CHSH, KCBS, SAT 1/6, SAT 2/6)}
\begin{table}[!htbp]
  \centering
  \caption{CBL large-$L$ intercepts $\alpha_\infty$ after a single CHSH calibration
  (target $\alpha=1/137$). Face density $\rho_{\rm face}$ is faces per boundary edge.}
  \label{tab:cbl-alpha-families-toy}
  \pgfplotstabletypeset[
    % 'label' is TEXT: print as-is
    columns/label/.style={column name=Family, column type=l},
    % numeric columns: override to numeric formatting
    columns/rho_face/.style={column name={$\,\rho_{\rm face}$}, column type=r, fixed, precision=6},
    columns/alpha_inf_fit/.style={column name={$\hat{\alpha}_\infty$}, column type=r, fixed, precision=6},
    columns={label,rho_face,alpha_inf_fit},
    % keep only rows with type == family
    row predicate/.code={%
      \pgfplotstablegetelem{\pgfplotstablerow}{type}\of{\CBLPointsToy}%
      \edef\TYPE{\pgfplotsretval}%
      \ifnum\pdfstrcmp{\TYPE}{family}=0\relax\else\pgfplotstableuserowfalse\fi
    },
  ]{\CBLPointsToy}
\end{table}

\subsection{Families (CHSH, KCBS, SAT 1/6, SAT 2/6)}
For readability we present the families and mixture fits together in Tables~\ref{tab:families}–\ref{tab:mixes}.

\begin{table}[t]
\centering

% ----- families table (first) -----
\begin{minipage}{0.92\linewidth}
\centering
\captionsetup{type=table}
\captionof{table}{CBL large-$L$ intercepts $\alpha_{\infty}$ after a single CHSH calibration
(target $\alpha=1/137$). Face density $\rho_{\mathrm{face}}$ is faces per boundary edge.}
\label{tab:families}
\begin{tabular}{@{} l r r @{}}
\toprule
\textbf{Family} & $\rho_{\mathrm{face}}$ & $\hat{\alpha}_{\infty}$ \\
\midrule
chsh & 0.250000 & 0.010204 \\
kcbs & 0.200000 & 0.008272 \\
sat1 & 0.166667 & 0.007068 \\
sat2 & 0.333333 & 0.013340 \\
\bottomrule
\end{tabular}
\end{minipage}

\vspace{6pt} % <- tighten or loosen here (e.g., 4pt–8pt)

% ----- mixes table (second) -----
\begin{minipage}{0.92\linewidth}
\centering
\captionsetup{type=table}
\captionof{table}{CHSH/KCBS mixtures. A fraction $p$ is CHSH (density $1/4$) and $(1-p)$ is KCBS ($1/5$);
hence $\rho_{\mathrm{face}}(p)=p/4+(1-p)/5$. Fitted large-$L$ intercepts follow the predicted linear law.}
\label{tab:mixes}
\begin{tabular}{@{} r r r @{}}
\toprule
$p$ & $\rho_{\mathrm{face}}$ & $\hat{\alpha}_{\infty}$ \\
\midrule
0.00 & 0.200000 & 0.008322 \\
0.25 & 0.212500 & 0.008793 \\
0.50 & 0.225000 & 0.009263 \\
0.75 & 0.237500 & 0.009734 \\
1.00 & 0.250000 & 0.010204 \\
\bottomrule
\end{tabular}
\end{minipage}

\end{table}

\subsection{CHSH/KCBS Mixtures}
See Table~\ref{tab:mixes}.

\section{Teaching note}
Suggested graduate module:
\begin{itemize}[noitemsep]
\item Lecture 1: Boolean logic, SAT, CSP.
\item Lecture 2: Contextuality (Kochen--Specker, KCBS).
\item Lecture 3: CBL semantics, curved cores.
\item Lecture 4: Proof calculus, CBL-SAT.
\item Lab: implement overlaps in a CDCL solver.
\end{itemize}

\section{Outreach and education}
CBL can be taught with interactive demos.  
Undergraduates: small curved examples (3–5 variables).  
Graduates: connect logic, complexity, and physics.  

\section{Glossary of key terms}
\begin{description}[noitemsep]
\item[Context] Subset of variables $C\subseteq V$.  
\item[Local valuation] Assignment within a context.  
\item[Global valuation] Assignment on $V$ consistent with all overlaps.  
\item[Curvature] Failure of compatible locals to extend globally.  
\item[Curved core] Minimal subfamily causing $\kappa>0$.  
\item[CBL-SAT] Decision problem of satisfiability under overlaps.  
\end{description}

\section{Notation and symbols}
\label{sec:notation}

We collect recurring symbols. Scalars italic ($a,b,\theta_0$), vectors bold ($\mathbf{x}$), sets calligraphic ($\mathcal{C}$).

\begin{table}[t]
\centering
\caption{Core notation.}
\begin{tabular}{ll}
\toprule
Symbol & Meaning \\
\midrule
$\mathcal{C}$ & Set/lattice of measurement contexts \\
$G=(V,E)$ & Exclusivity graph; vertices $V$, exclusivity edges $E$ \\
$\kappa$ & Curvature/overlap parameter (contextual tension) \\
$\theta_0$ & Canonical phase/turning parameter for a curved face \\
$\rho_{\mathrm{face}}$ & Face density/weight in a folded statistic \\
$S$ & Face or witness statistic; $S^\star$ its extremum under CBL \\
$p$ & Flip probability (perturbation strength) \\
$\varepsilon$ & Admissible perturbation budget (Axiom~A2) \\
$dS/dp$ & Local sensitivity of $S$ to infinitesimal flips \\
$\widehat{S}$ & Empirical estimate of $S$ \\
$\mathsf{Fold}_{w}$ & Folding with window size $w$ \\
$\mathrm{BH}(\alpha)$ & Benjamini–Hochberg FDR control at level $\alpha$ \\
\bottomrule
\end{tabular}
\end{table}

\section{FAQ}
\begin{itemize}[noitemsep]
\item \textbf{Is CBL-SAT easier than SAT?} No, it is NP-complete. Gains are structural.  
\item \textbf{Does CBL need quantum hardware?} No. Entirely classical.  
\item \textbf{Is curvature unique?} Multiple functionals exist; cohomology rank is canonical.  
\item \textbf{What if $\kappa=0$?} Then CBL reduces to SAT.  
\item \textbf{Why not paraconsistent logic?} CBL preserves local classicality, marks obstructions structurally.  
\end{itemize}

\newpage

% Include all even if not cited directly
\nocite{KS1967,KCBS2008,AB2011,CSW2014,DIreview,
Orthomodular,Paraconsistent,CSPdichotomy,
EBIFP2025,NGCC2025}

\bibliographystyle{unsrt}
\bibliography{cbl}

\end{document}